\documentclass[10pt,journal,compsoc]{IEEEtran}

\ifCLASSOPTIONcompsoc
  \usepackage[nocompress]{cite}
\else
  \usepackage{cite}
\fi

\usepackage{amsmath,amssymb,amsfonts}

\usepackage{textcomp}
\usepackage{xcolor}
\usepackage{algorithm}
\usepackage{algorithmicx}
\usepackage[noend]{algpseudocode}
\usepackage{amsmath}
\usepackage{capt-of}
\usepackage{wrapfig}
\usepackage{multirow}
\usepackage{booktabs}
\usepackage{subfig}
\usepackage{url}
\usepackage{diagbox}
\usepackage{threeparttable}
\usepackage{blkarray}

\ifCLASSINFOpdf
  \usepackage[pdftex]{graphicx}
\else
  \usepackage[dvips]{graphicx}
\fi

\usepackage{adjustbox}

\newcolumntype{P}[1]{>{\centering\arraybackslash}p{#1}}

\newcommand{\Continue}{\State \textbf{continue}}

\DeclareMathOperator*{\argmin}{arg\,min}

\newtheorem{theorem}{Theorem}

\newtheorem{definition}{Definition}
\newtheorem{proof}{Proof}[section]

\begin{document}

\title{PICO: Pipeline Inference Framework for Versatile CNNs on Diverse Mobile Devices}

\author{Xiang~Yang,
	Zikang~Xu,
	Qi~Qi,
	Jingyu~Wang,
	Haifeng~Sun,
	Jianxin~Liao,
	and~Song~Guo,~\IEEEmembership{Fellow,~IEEE}
	\IEEEcompsocitemizethanks{
		\IEEEcompsocthanksitem Xiang Yang, Zikang Xu, Qi Qi, Jingyu Wang, Haifeng Sun and Jianxin Liao are with the State Key Laboratory of Networking and Switching Technology, Beijing University of Posts and Telecommunications. 
		E-mail: \{yangxiang, xuzikang, qiqi8266, wangjingyu, hfsun, liaojx\}@bupt.edu.cn.  
		\IEEEcompsocthanksitem Song Guo is an IEEE Fellow (Computer Society) and an ACM Distinguished Member with the Department of Computing at The Hong Kong Polytechnic University.
		E-mail: cssongguo@comp.polyu.edu.hk
		\IEEEcompsocthanksitem Qi Qi and Jingyu Wang are the corresponding authors.}
}

\IEEEtitleabstractindextext{%
	\begin{abstract}
		Distributing the inference of convolutional neural network (CNN) to multiple mobile devices has been studied in recent years to achieve real-time inference without losing accuracy. However, how to map CNN to devices remains a challenge. On the one hand, scheduling the workload of state-of-the-art CNNs with multiple devices is NP-Hard because the structures of CNNs are directed acyclic graphs (DAG) rather than simple chains. On the other hand, distributing the inference workload suffers from expensive communication and unbalanced computation due to the wireless environment and heterogeneous devices.
		This paper presents PICO, a pipeline cooperation framework to accelerate the inference of versatile CNNs on diverse mobile devices. At its core, PICO features: (1) a generic graph partition algorithm that considers the characteristics of any given CNN and orchestrates it into a list of model pieces with suitable granularity, and (2) a many-to-many mapping algorithm that produces the best pipeline configuration for heterogeneous devices.  In our experiment with $2 \sim 8$ Raspberry-Pi devices, the throughput can be improved by $1.8 \sim 6.8 \times$ under different CPU frequencies.
		
	\end{abstract}

	}

\maketitle

\IEEEdisplaynontitleabstractindextext

\IEEEpeerreviewmaketitle

\IEEEraisesectionheading{\section{Introduction}\label{sec:introduction}}

\begin{figure*}
	\centering
	\includegraphics[width=0.85\linewidth]{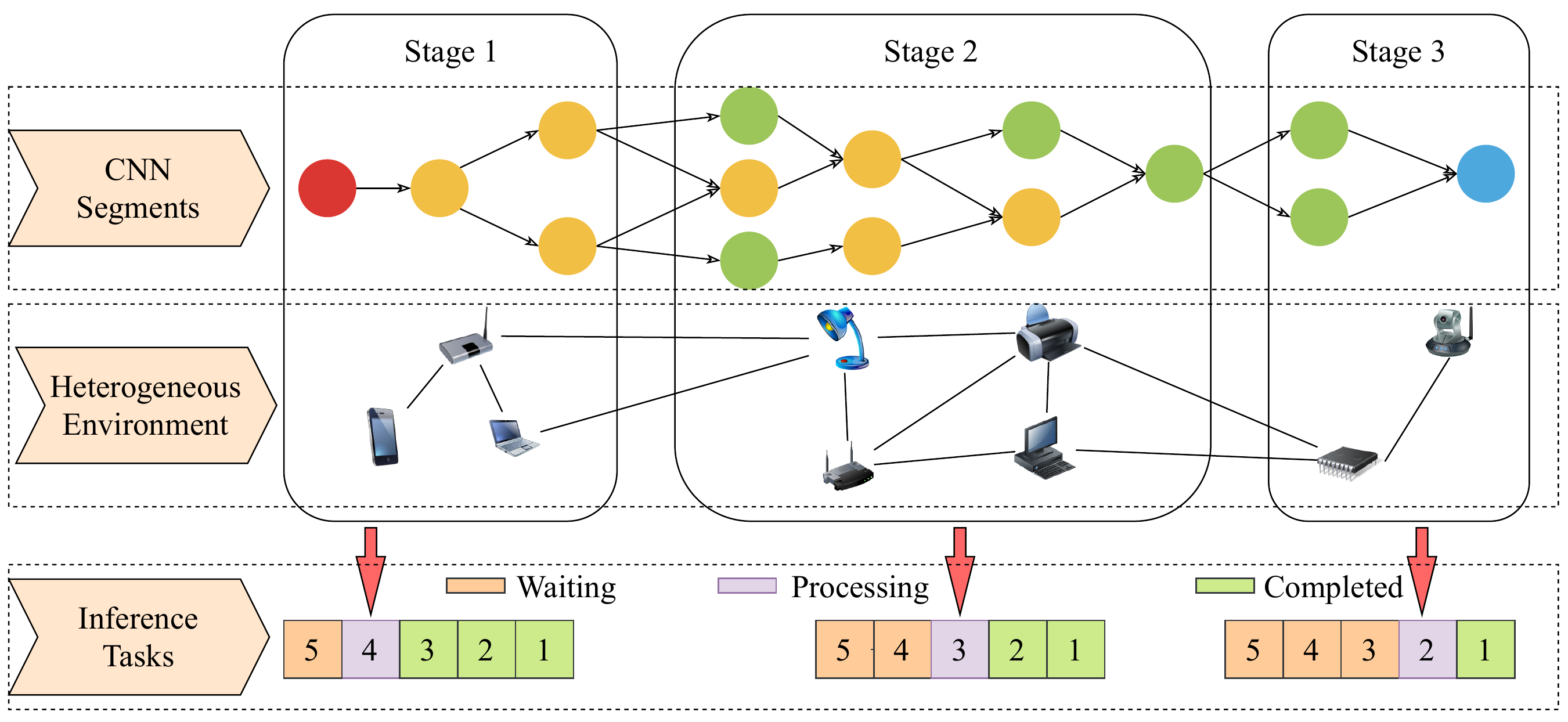}
	\caption{A diagrammatic sketch of pipeline inference.}
	\label{fg:pipeline-overview}
\end{figure*}

\IEEEPARstart{R}{ecent} years witness an explosive growth of mobile devices. The huge number of mobile devices provides a large volume of data (images, videos, etc.). Meanwhile, versatile convolution neural networks (CNN) with pre-trained parameters become powerful tools to make intelligent decisions using these raw data (\textit{CNN inference}). 
Embedding CNN with mobile devices enables many intelligent applications to become reality, such as smart home, intelligent factory, and even automatic driving \cite{kour2014real,kour2014fast}.

One obstacle to the embedding is the resource-limited mobile devices. Compared with datacenter, the computing capability of mobile devices is not enough to perform CNN inference alone. But on the contrary, the current wireless network is not prepared for transmitting the massive volume of raw data collected by these mobile devices. 
For example, an autopilot camera could capture more than 700 MB video record every second \cite{zhou2019edge}, and uploading all the video data to the datacenter will bring significant network latency. Moreover, uploading data from user devices to the cloud always brings concern about privacy \cite{zhou2019edge}.

Benefitted from the spatial independence of convolution operation, the input and output (\textit{feature}) of convolutional (\textit{conv}) layers can be split into several small tiles and executed on different devices. 
As a consequence, cooperative CNN inference on multiple mobile devices gains the attention of researchers recently \cite{mao2017modnn,zhao2018deepthings,zhou2019adaptive}. 
During inference procedure, the data source (camera, sensors, etc.) captures raw data and splits it into tiles. These tiles are distributed to multiple nearby mobile devices through a wireless network and executed independently using one or several layers. Then the data source is responsible to gather all the output tiles and stitch them to obtain the result. The procedure will be iterated multiple times until all layers are executed.
Cooperative inference also protects user privacy since all the data stay in local. Moreover, the closer to the data source, the lower network latency it suffers.
Compared with other strategies such as model compression and parameter pruning \cite{li2019oicsr,he2019filter,howard2019searching}, cooperative inference neither losses the inference accuracy nor requires re-train the model.

However, despite all these benefits, there still leave some challenges that are not completely solved in previous works.
Although the input feature can be parallel executed, (1) \textbf{the parallelism introduces redundant calculation} due to the property of CNN. 
The scalar in the output feature of one conv layer is calculated through a dot product with the conv kernel and a subregion of input feature. For most cases, the kernel size is bigger than $1 \times 1$, so that the input tiles of partitioned input feature will overlap with each other to guarantee the scalars at the edge of output tiles are correct.
Moreover, the overlapped part will increase recursively when devices execute multiple layers during one iteration in the inference procedure, but the communication is expensive for mobile devices. As a consequence, the executed layers need to be carefully chosen.
However, (2) \textbf{the structures of many CNNs are directed acyclic graphs} (DAG) rather than chains. ResNet34 \cite{resnet} uses skip-connection technology that allows a layer to directly connects to another deeper layer. The structure of InceptionV3 \cite{szegedy2017inception} contains multiple branches to capture more information from the input feature. These complex structures lead to a huge number of possible choices. Previous works mainly focused on the chain structure \cite{mao2017modnn,zhao2018deepthings,zhou2019adaptive}, which is much easier than DAG. 
Compared with datacenter, (3) \textbf{the computing resources of mobile devices are diverse}, the heterogeneous environment also hinders the optimization for cooperative inference.

In this paper, we explore previous works about parallelizing the CNN inference and propose a pipeline cooperation (PICO) framework for accelerating the inference on diverse mobile devices. Fig. \ref{fg:pipeline-overview} plots a diagrammatic sketch of our framework. PICO divides the entire CNN graph and mobile devices into $3$ \textit{stages}. These stages compose an efficient inference pipeline. Since each stage owns a small segment of original CNN and a subset of mobile devices, both communication overhead and the redundant calculation can be significantly reduced. There are two import metrics for pipeline: \textit{latency} and \textit{period}. The first term is the sum of inference latencies of all stages and the last term is the longest latency among stages. Obviously decreasing the period tends to increase the latency. Our optimization goal is to minimize the redundant calculation and period (maximize throughput) meanwhile to keep the latency of the pipeline under a certain value.

We first formulate the pipeline inference, then we analyze the complexity of the optimization problem and find that it is NP-Hard to directly obtain the optimal result. Based on our analysis, PICO uses a two-step optimization to maximize the throughput. In the first step, we orchestrate the CNN graph into a sequence of \textit{pieces}. These pieces have the minimum redundant calculation inside and compose the original CNN graph in a chain structure. Then we choose the best partition for these pieces and devices to construct the inference pipeline. The algorithms used in the above procedures are based on dynamic programming.

In our experiment we use $2 \sim 8$ Raspberry-Pi devices to evaluate PICO framework. The throughput can be improved by $1.8 \sim 6.8 \times$ under different CPU frequencies and number of devices.

In a nutshell, we make the following contributions:
\begin{itemize}
	\item We present a pipeline cooperation (PICO) framework to accelerate CNN inference with diverse mobile devices.
	\item We propose an algorithm to split the complex CNN graph structure into more fine-grained pieces.
	\item We propose an algorithm to decide the optimal stage settings for inference pipeline which maximize the throughput.
	\item We apply our technique on a cluster consisting of Raspberry-Pi-based hardware and evaluate image recognition and object detection CNN models.
\end{itemize}

The rest of this paper is organized as follows: Section \ref{sec:background} provides background information of CNN and different parallelization strategies in mobile devices. Section \ref{sec:model} formulates the inference process and gives a cost model for further optimization. Section \ref{sec:graph} and \ref{sec:main-algo} describe our approach to find near-optimal parallelization. Section \ref{sec:experiment} presents the results of our evaluation. Section \ref{sec:related} details the related work and Section \ref{sec:conclusion} concludes.

\section{Background And Motivation} \label{sec:background}
In this section, we briefly introduce the CNN inference and the current parallel schemes. Then we propose our pipeline cooperation scheme.

\subsection{Procedure of CNN Inference}

The convolution layer (\textit{conv}) is the key module during CNN inference, Each conv layer owns a set of kernels. To produces the output feature, conv layers use their kernels to slide over the input feature received from the previous layer. 
Every movement of the kernel will produce a scalar in the output feature by a dot product between the weights of kernels and a small subregion of the input.
The pooling layer (\textit{pool}) performs a down-sampling operation. It is used to progressively reduce the number of parameters, memory footprint and amount of computation in the network.

Conv operation is the biggest bottleneck. Fig. \ref{fg:layer-overhead} plots the computation and communication percentage by layer for two classic CNNs (VGG16 \cite{simonyan2014very}, YOLOv2 \cite{redmon2017yolo9000}). 
From the figure we can find conv layers dominate the consumption of computing resources. The conv operations occupy $99.19\%$ of the computation in VGG16 and $99.59\%$ in YOLOv2. How to efficiently execute conv operations is the key to accelerating CNN inference.
Another finding is the variation. 
Since different conv layers have different configurations (kernel size, padding, in and out channels), the communication or computation percentage also varies. 

\begin{figure}
	\centering
	\subfloat[VGG16]{ \label{fg:vgg16-overhead}
		\includegraphics[width=0.8\linewidth]{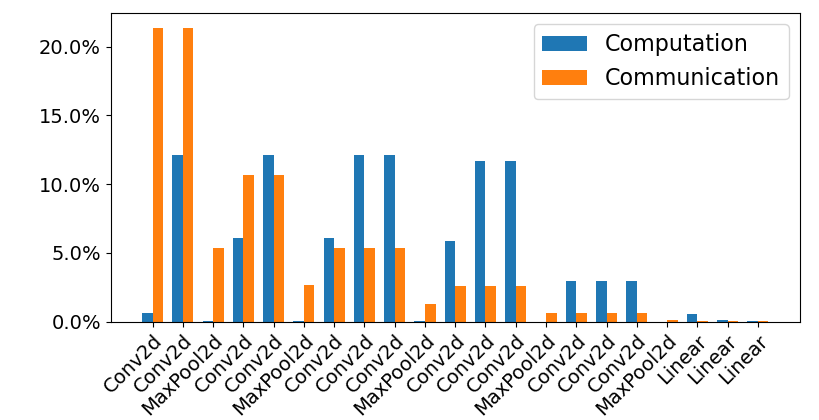}
	}
	
	\subfloat[YOLOv2]{ \label{fg:yolov2-overhead}
		\includegraphics[width=0.8\linewidth]{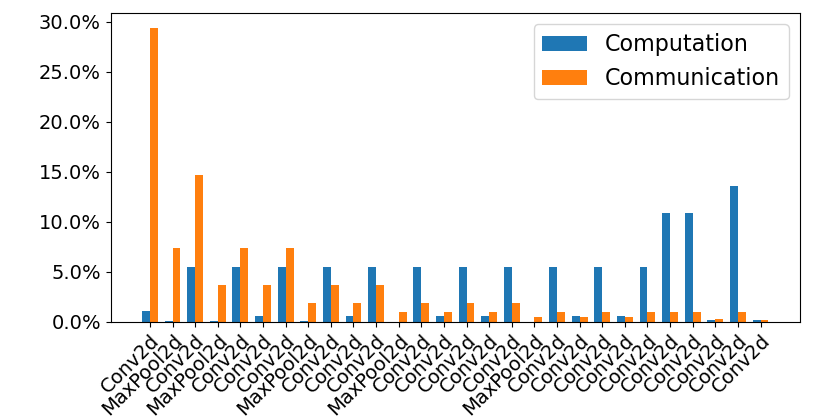}
	}
	\caption{The communication and computation percentages of each layer.}
	\label{fg:layer-overhead}
\end{figure}

\subsection{Parallelizing CNNs With Mobile Devices}
Benefitted from the spatial independence of conv operations, the inference can be parallel executed by splitting the input feature into multiple tiles and distributing them to different mobile devices, as shown in Fig. \ref{fg:feature-split}. We refer this technology as \textit{feature partition}.
However, the partition of input feature overlaps with each other due to the property of conv operations.
In Fig. \ref{fg:feature-split}, an input feature is split into four tiles and distributed to four devices. Assume the corresponding conv layer has a $3 \times 3$ kernel size,
to obtain the correct value in $P_1$, the calculation with $3 \times 3$ kernel has to use more proportion (the edges of the yellow and pink region) of the input feature. This property leads to a \textit{redundant calculation} and increases the difficulty of the design of parallel algorithm.

We next introduce the two parallelization schemes used in this paper. \cite{mao2017modnn} is the first work that uses feature partition for cooperative CNN inference. 
For each layer, the basic idea is to split the input feature into tiles and distribute them to all devices, then gather them to obtain the output of this layer. We refer such a scheme as \textit{layer-wise} parallelization.
In a WLAN network, it can cause substantial network latency. The gain of layer-wise parallelization is significantly defeated by communication overhead. 
To reduce the communication among devices, \textit{fused-layer} parallelization was introduced in \cite{zhao2018deepthings} and \cite{zhou2019adaptive}. This scheme fuses multiple layers instead of distributing the computation of every layer individually. Thus, mobile devices can execute the calculation of multiple layers without communication. 
But since the input will go through multiple layers, to obtain the correct value of output feature, the overlapped part of the input increases recursively.
In addition, all mobile devices need a full copy of original CNN for the two schemes, which increases the memory footprint.

\subsection{The Structures Of CNNs}

\begin{figure*}
	\begin{minipage}[c]{0.58\textwidth}
		\subfloat[Chain Structure (\textit{VGG16})]{ \label{fg:dnn-struct-chain}
			\includegraphics[width=0.98\linewidth]{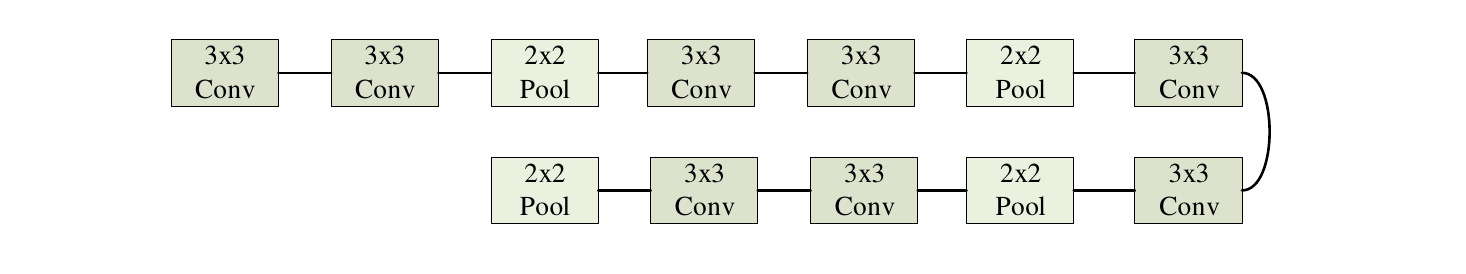}
		}
		
		\subfloat[Block Structure (\textit{InceptionV3})]{ \label{fg:dnn-struct-block}
			\includegraphics[width=0.98\linewidth]{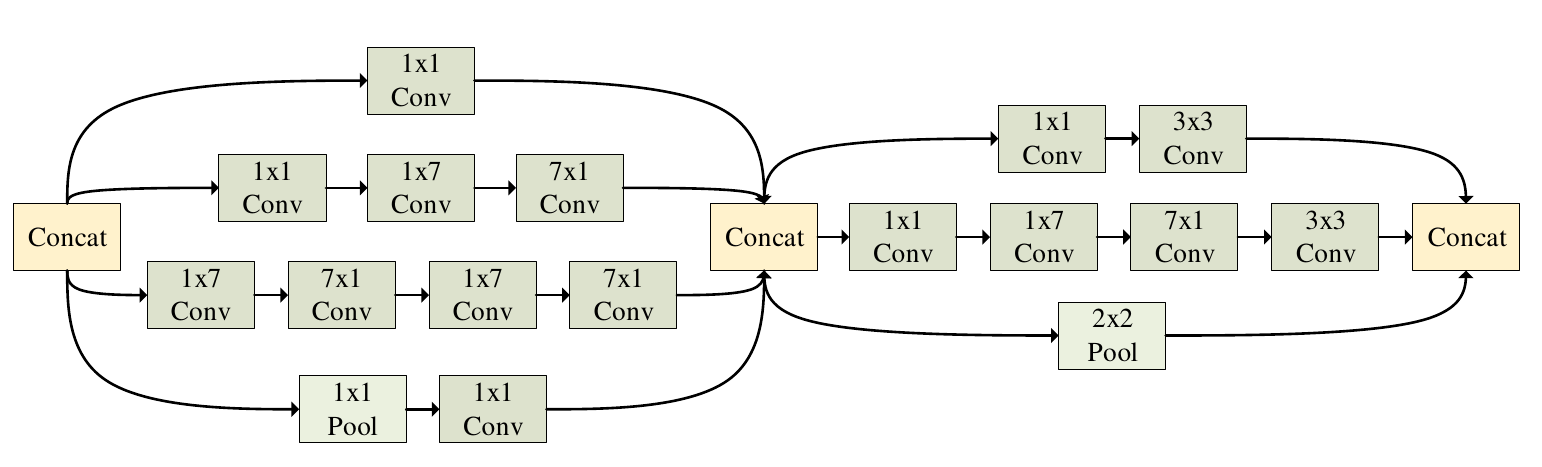}
		}
	\end{minipage}
	\begin{minipage}[c]{0.38\textwidth}
		\subfloat[Graph Structure (\textit{NasNetMobile})]{ \label{fg:dnn-struct-graph}
			\includegraphics[width=0.98\linewidth]{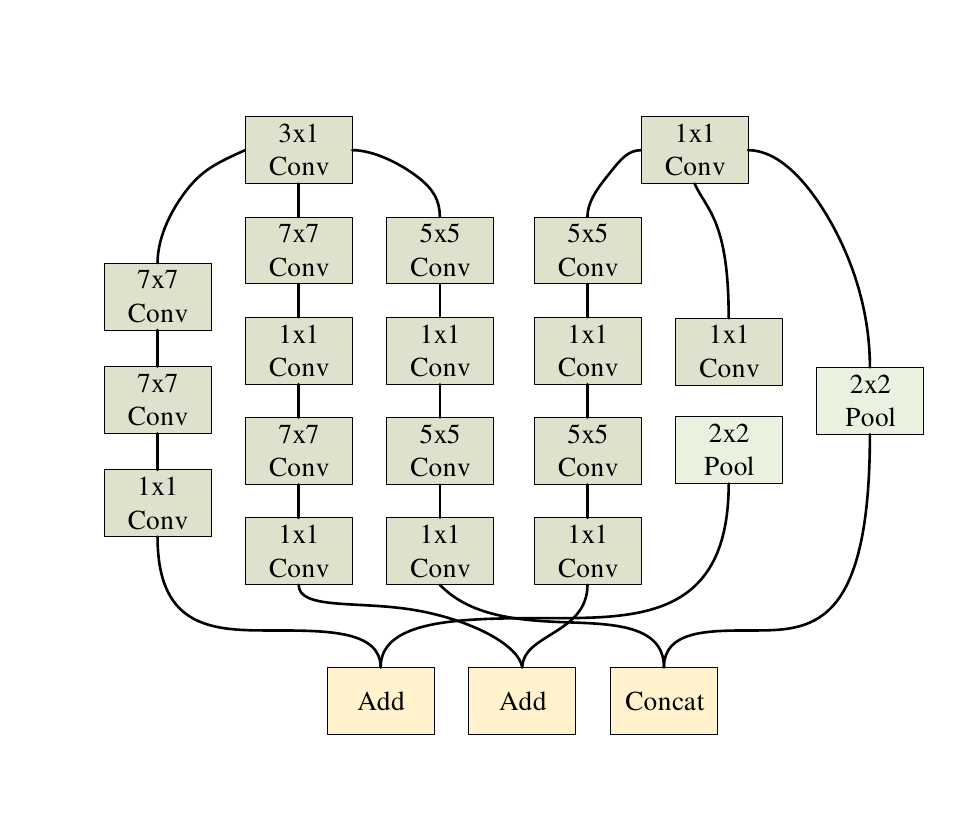}
		}
	\end{minipage}
	\caption{CNNs with different structures:
		The chain structure is the simplest one which just put the neural layers into a sequence.
		Block structure replaces the element in chain structure from neural layer to block, each block can be seen as a directed acyclic graph (DAG).
		Graph structure can not be partitioned into blocks, the entire structure is a huge DAG.
	}
	\label{fg:dnn-struct}
\end{figure*}

The structures of CNNs can be divided into three categories. We plot Fig. \ref{fg:dnn-struct} to give an illustration. Note the norm layer and activation layer are ignored since they do not change the input and output shape and has less proportion of computation.

The earlier model such as VGG16 \cite{simonyan2014very} and YOLOv2 \cite{redmon2017yolo9000} are built with the (1) \textbf{chain} structures. 
This structure is simple: neural layers inside the model are connected one by one, and the output of the previous layer is the input of the next layer. We plot the model structure of VGG16 in Fig. \ref{fg:dnn-struct-chain} for further explanation.

Later, the (2) \textbf{block} structure becomes popular in CNNs. Block structure enables CNNs to capture multiple features of input data to improve its performance using carefully designed blocks \cite{szegedy2017inception} and prevent the vanishing gradient problem when training deeper model \cite{resnet}. 
It uses blocks to replace the layers in chain structure. All the blocks are still connected one by one, but inside the block, neural layers can be represented as an acyclic directed graph (DAG). Fig. \ref{fg:dnn-struct-block} plots the 8th and 9th blocks in InceptionV3 \cite{szegedy2017inception}. Each block has multiple branches and contains several conv and pool layers, and these blocks are connected with the \text{Contact} operations that stacks the output of every sink layer of previous block in channel dimension and feeds the result to the next block.

To avoid manual design of the model structure, neural architecture search (NAS) is proposed. Compared with the previous two structures, the output structure of NAS is usually a complete graph, which can not be divided into sequence of blocks.
We refer the structure as (3) \textbf{graph} structure. We plot a partition of NasNetMobile \cite{cvpr18:nashnet} in Fig. \ref{fg:dnn-struct-graph}, which has two source layers and three sink layers. The complex structure of CNN models is a big challenge for optimizing parallel strategy.

\subsection{Motivation And Pipeline Inference}

\begin{figure}
	\centering
	\includegraphics[width=0.88\linewidth]{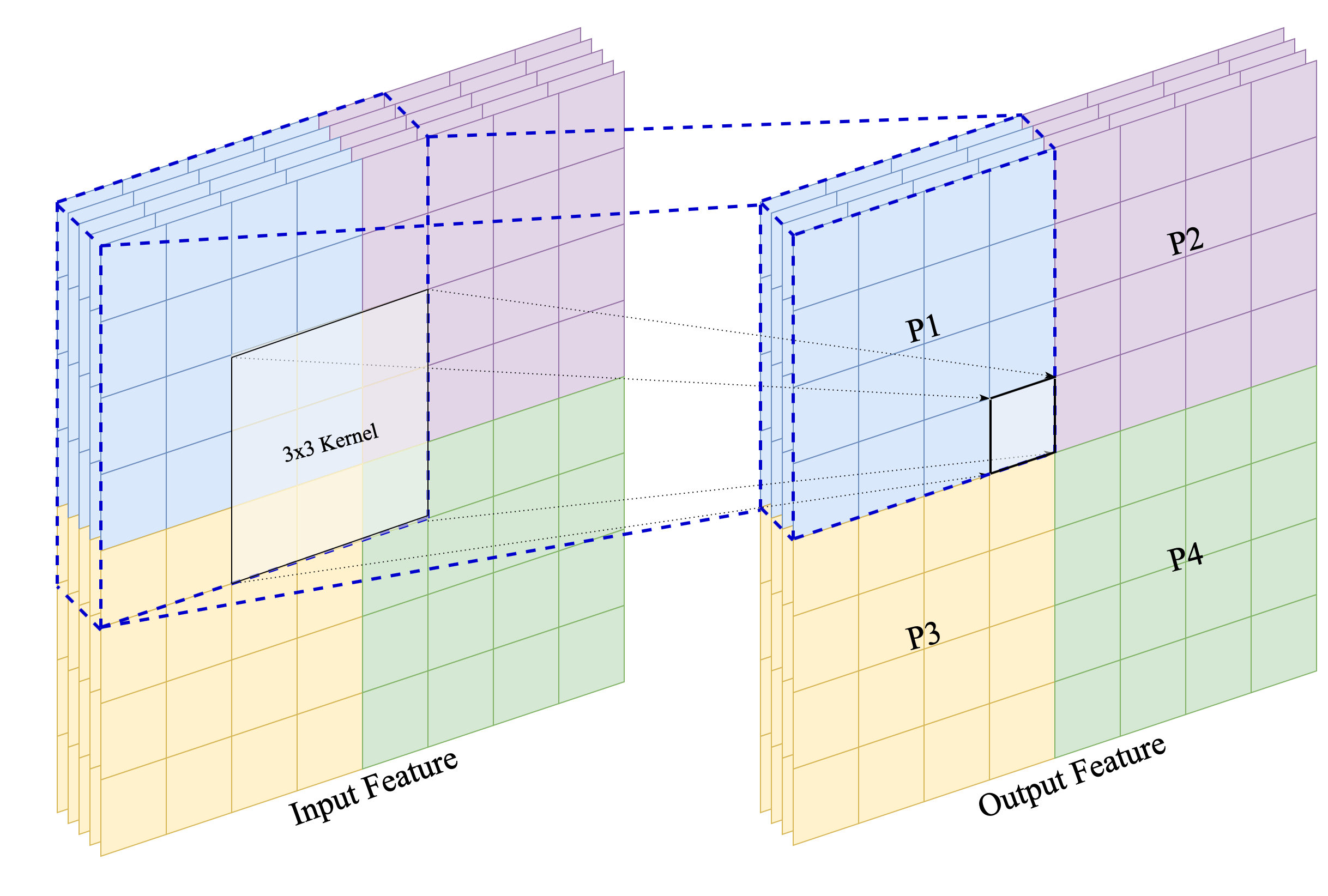}
	\caption{Feature map partition strategy introduces redundant calculation.}
	\label{fg:feature-split}
\end{figure}

\begin{figure}[tb]
	\centering
	
	\subfloat[Per device overhead]{\includegraphics[width=0.44\linewidth]{
			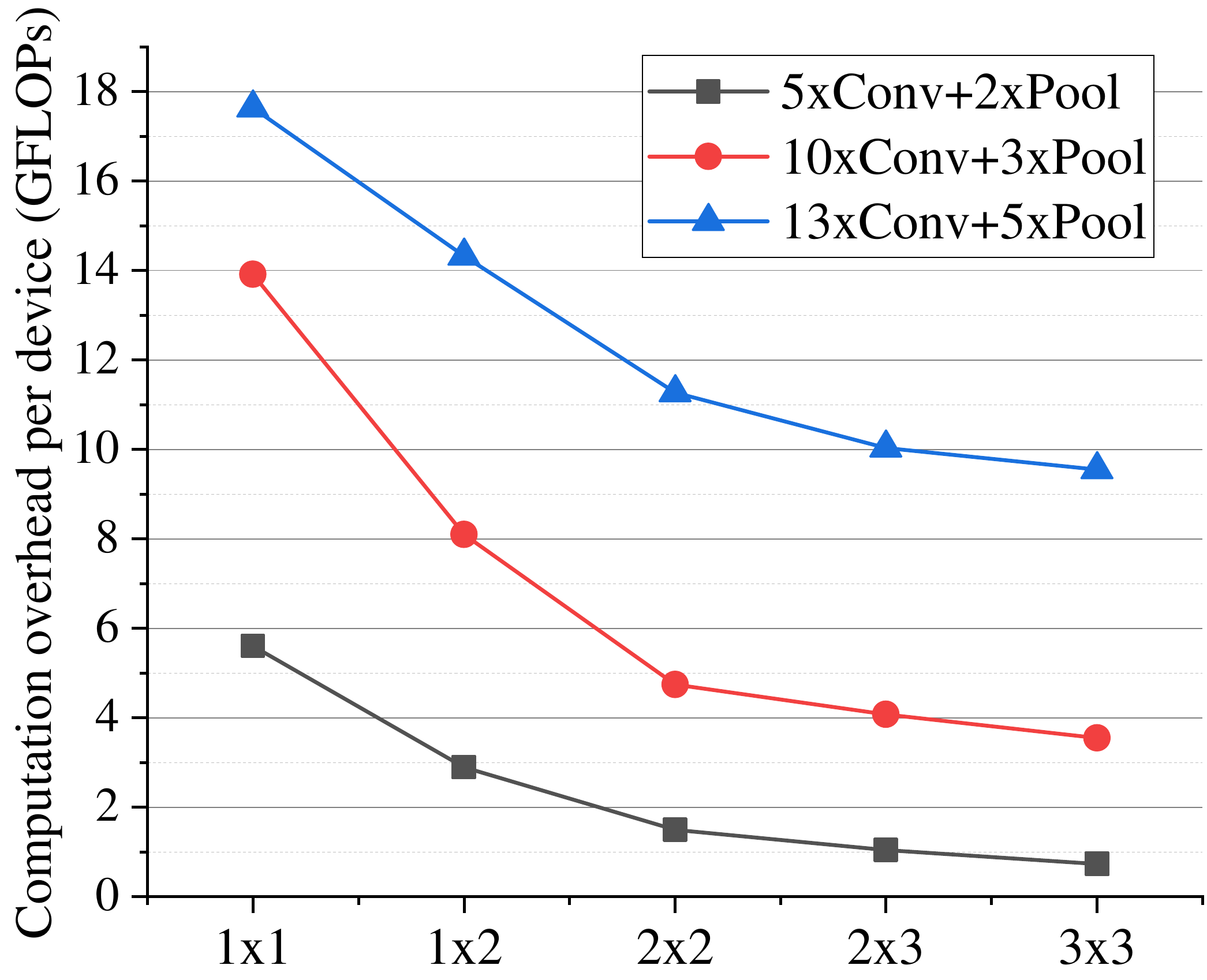}  \label{fg:computation-per}
	}
	\subfloat[Total computation overhead]{\includegraphics[width=0.46\linewidth]{
			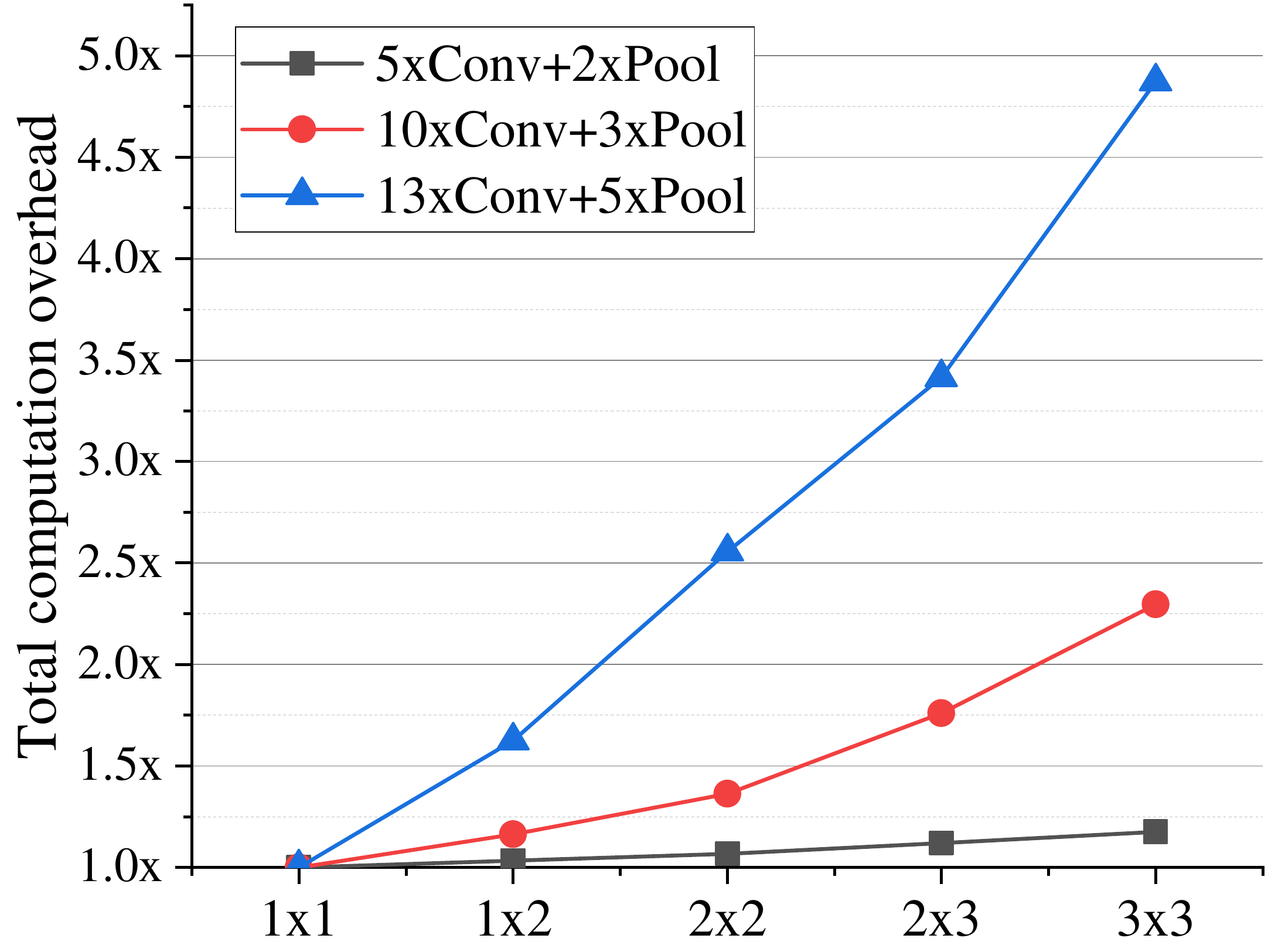} \label{fg:computation-total}
	}
	\caption{Computation overhead with different partition settings.}
	\label{fg:computation}
\end{figure}

\subsubsection{Motivation}
In the above discussion, How to tackle the complex graph structure of CNN model and how to reduce the redundant calculation are keys to accelerating the CNN inference.

\textbf{Tackle the complex structure:}
Cooperative inference needs to distribute the CNN model into multiple devices, but the structure of these model is complex and prevent more fine-grained optimization. Lots of previous works focus \cite{mao2017modnn,zhao2018deepthings,zhou2019adaptive} on cooperative inference, they only consider chain structure. There lacks a solution for the more complex block and graph structures.

\textbf{Reduce the redundant calculation:}
Communication is expensive in mobile environment. For layer-wise scheme, frequent communication among mobile devices causes inefficient performance.
The redundant calculation also limits the cooperation of mobile devices for CNN inference. 
For fused-layer scheme, the redundant calculation quickly grows as the number of fused layers or devices increases.
To give a detailed explanation, we evaluate the required floating-point operations (FLOPs) for VGG16 with different numbers of fused layers and mobile devices.
Fig. \ref{fg:computation-per} presents the FLOPs per device meanwhile Fig. \ref{fg:computation-total} shows the sum of FLOPs of all devices.
We can find that fused-layer strategy performs well at the start, but when the numbers of fused layers or devices increase, the redundant calculation quickly grows.

\subsubsection{Pipeline Inference}
From the above discussion, the acceleration of CNN inference faces challenges when the number of devices or fused layers grows.
For layer-wise scheme, the devices are idle in most time due to frequent communication and expensive network latency. 
On the contrary, devices can keep running with the fused-layer scheme, but it is whistling to the wind since the most computation is redundant.
Here we introduce the \textit{pipeline} scheme for parallelizing the CNN inference. 
This scheme divides both layers and mobile devices into several groups, as shown in Fig. \ref{fg:pipeline-overview}. We refer such a group as \emph{stage} in our description. 
The inference inside the stage uses fused-layer scheme and the entire CNN inference is performed stage by stage. If we set the number of stages to 1, The fused-layer scheme is a special case of the pipeline scheme.
To maximize the inference throughput, the inference latency of every stage should be optimized as close as possible.

Using pipeline for inference has several advantages. (1) First, the communication overhead can be reduced since the calculated features only need to be synchronized among a subset of devices. 
(2) Second, the proportion of redundant calculation also decreases due to smaller numbers of layers and devices. (3) Third, each device owns a segment of CNN instead of the entire model, which reduces the memory footprint.

The concept of pipeline is widely adopted in task scheduling \cite{benoit2008mapping,benoit2010complexity} which maps multiple processors to an application composed of several tasks.
However, pipeline meets difficulties when applied to CNN inference. The structure of CNN is a directed acyclic graph (DAG) rather than a chain, the mapping has to consider the data flow of DAG. 
Generally, the number of layers in CNN is more than the number of devices, thus the mapping is many-to-many, and different mapping strategy also changes redundant calculation.
Moreover, the heterogeneous environment is also a big challenge.

\section{System Model} \label{sec:model}

\begin{table}
	\caption{Notation definitions}
	\centering
	\begin{tabular}{|l|l|}
		\toprule
		\textbf{Notation}                                 & \textbf{Description}                                        \\
		\midrule
		$\mathbb{G}:(\mathbb{V}, \mathbb{E})$             & CNN with graph structure.                                   \\
		$\mathbb{D}$                                      & A heterogeneous cluster with $D$ devices.                   \\
		$l_i$                                             & Layer $i$ in model $\mathbb{G}$.                            \\
		$w_i$, $h_i$                                      & The width and height of the output frame of $l_i$.          \\
		$k_i, p_i, s_i, c_i$                              & Kernel size, padding, stride, and channel of $l_i$.         \\
		$d_k$                                             & A device in cluster $\mathbb{D}$.                           \\
		$F_i^k$                                           & Input feature frame of layer $l_i$ for device $d_k$.        \\
		$\mathcal{F}^k$                                   & Set of all input feature for layers assigned to $d_k$.      \\
		$\mathcal{F}_{in}^k$                              & Input feature for source layers assigned to $d_k$.          \\
		$\mathcal{F}_{out}^k$                             & Output feature for sink layers assigned to $d_k$.           \\
		$b(d_h, d_k)$                                     & Bandwidth between device $d_h$ and device $d_k$.            \\
		$\mathcal{D}$                                     & A subset of devices.                                        \\
		$\mathcal{M}:(\mathcal{V}, \mathcal{E})$          & Model partition deployed on $d_k \in \mathcal{D}$.          \\
		$\mathcal{S}:(\mathcal{M}, \mathcal{D})$          & A stage that belongs to the inference pipeline.             \\
		$\mathcal{M}_E$                                   & Ending piece of CNN $\mathbb{G}$.                           \\
		$\varphi(\mathcal{F}_k)$                          & Input frame size of $F_k$.                                  \\
		$\theta(\mathcal{M};\mathcal{F}^k)$               & Required computing resources of $\mathcal{M}$.              \\
		$\vartheta(d_k)$                                  & Computing capacity of device $d_k$.                         \\
		$t_{comm}(d_k, \mathcal{F}^k)$                    & Communication time of device $d_k$.                         \\
		$t_{comp}(d_f,\!d_k,\!\mathcal{F}^k)$             & Computation time of device $d_k$.                           \\
		$T({\mathcal{S}})$                                & Time overhead for executing stage $\mathcal{S}$.            \\
		$T_{comm}(\mathcal{S})$                           & Communication time of stage $\mathcal{S}$.                  \\
		$T_{comp}(\mathcal{S})$                           & Computation time of stage $\mathcal{S}$.                    \\
		$T_{lim}$                                         & Inference latency limit for optimization.                   \\
		$\mathbb{S}$                                      & Pipeline configuration containing all stages $\mathcal{S}$. \\
		$\mathbb{S}^\star$                                & Optimal stage configuration.                                \\
		$\mathcal{T}(\mathbb{G}, \mathbb{D}, \mathbb{S})$ & Latency of the pipeline under configuration $\mathbb{S}$.   \\
		$\mathcal{P}(\mathbb{G}, \mathbb{D}, \mathbb{S})$ & Period of the pipeline under configuration $\mathbb{S}$.    \\
		\bottomrule
	\end{tabular}
	\label{tab:symbol}
\end{table}

In this section, we define our optimization problem for pipeline inference.

\subsection{Problem Define}
Generally speaking, our goal is to divide both CNN model with graph structure and mobile devices with heterogeneous computing resources into several stages properly, so that these stages could compose an inference pipeline that maximizes the throughput.

\subsubsection{CNN With Graph Structure}
We use an acyclic directed graph (DAG) $\mathbb{G}: <\mathbb{V}, \mathbb{E}>$ to represent a given CNN model. The  vertex set $\mathbb{V}$ contains all the neural layers and connector (e.g., Add and Contact in Fig. \ref{fg:dnn-struct}) $l_i \in \mathbb{V}$, and the elements $(l_i, l_j)$ in the edge set $\mathbb{E}$ denotes the data flow of CNN model $\mathbb{G}$. In particular, $(l_i, l_j) \in \mathcal{E}$ means the output of layer $l_i$ is the input of layer $l_j$.
Since the CNN model will be executed as an inference pipeline with multiple stages, the $\mathbb{G}$ also needs to be split into multiple parts. We refer these parts as \textit{segments}. A segment  $\mathcal{M} : <\mathcal{V}, \mathcal{E}>$ is a subset of original DAG $\mathbb{G}$, where $\mathcal{V} \subseteq \mathbb{V}$ and $\mathcal{E} \subseteq \mathbb{E}$. 

Note the segment \textit{is not} a regular smaller graph, since the edge set $\mathcal{E}$ contains some vertices that are not included by $\mathcal{V}$. Take Fig. \ref{fg:pipeline-overview} as an example, these segments on the top also contain edges that are connected with previous or next segments.
Here we give some definitions of segment to simply our following modeling:
\begin{definition}
	A subset $\mathcal{M} : <\mathcal{V}, \mathcal{E}>$ is a \textit{segment} of original graph $\mathbb{G}: <\mathbb{V}, \mathbb{E}>$ if for all $e:(u, v) \in \mathbb{E}$, once $u$ or $v$ belongs to $\mathcal{V}$, $e$ also belongs to $\mathcal{E}$.
\end{definition}
\begin{definition}
	For a segment $\mathcal{M} : <\mathcal{V}, \mathcal{E}>$ and an edge $(u, v) \in \mathcal{E}$, if $u \notin \mathcal{V}$, then $v$ is a \textit{source} vertex of $\mathcal{M}$.
\end{definition}
\begin{definition}
	For a segment $\mathcal{M} : <\mathcal{V}, \mathcal{E}>$ and an edge $(u, v) \in \mathcal{E}$, if $v \notin \mathcal{V}$, then $u$ is a \textit{sink} vertex of $\mathcal{M}$.
\end{definition}

\subsubsection{Optimization Goal}
Given a heterogeneous cluster $\mathbb{D}$, where $d_k \in \mathbb{D}$ is a computing device in the cluster.
We assume the computing capacity $\vartheta(d_k)$ of device $d_k$ are known. In our practice, the $\vartheta(d_k)$ denotes floating point operations per second (FLOPS).
We also assume the bandwidth between all mobile devices is the same and is known as $b$. This assumption covers most cases when these devices are under the same WLAN environment such as home and factory \cite{benoit2008mapping,zhou2019adaptive}.

For pipeline scheme, $\mathcal{D} \subseteq \mathbb{D} $ is a subset of heterogeneous devices. Each device $d_k \in \mathcal{D}$ owns a copy of model segment $\mathcal{M}$ but is assigned to produce different region  $\mathcal{F}^k$ of the output feature map of all the sink vertex in $\mathcal{M}$. We use $\mathcal{F}$ to present the set of all $\mathcal{F}^k$ in $\mathcal{D}$. A stage $\mathcal{S}$ can be represented as a tuple $(\mathcal{M},\mathcal{D}, \mathcal{F})$. Let $\mathbb{S}$ denote the set of stages composed by all the stages $\mathcal{S}$ we defined above, the optimization objective is to find such a $\mathbb{S}^\star$ that satisfies:
\begin{equation}
	\mathbb{S}^\star = \argmin_{\mathcal{T}(\mathbb{G}, \mathbb{D}, \mathbb{S}) \leq T_{lim}}
	\mathcal{P}(\mathbb{G}, \mathbb{D}, \mathbb{S})
\end{equation}
where $\mathcal{T}(\mathbb{G}, \mathbb{D}, \mathbb{S})$ denotes the pipeline latency under specific stage configuration $\mathbb{S}$ and $\mathcal{P}(\mathbb{G}, \mathbb{D}, \mathbb{S})$ is the period of pipeline. $T_{lim}$ is a hyperparameter that indicates the maximum inference latency we can accept.

\subsection{Cost Model} \label{sec:latency_model}

Here we represent our cost model to guide the optimization. First, we quantify the essential input feature size for every device in a stage. Then, we formulate the inference latency of every stage. Finally, we get the inference period and latency of the entire pipeline using previous results.

\subsubsection{The Input Feature Size For Devices}
Every device $d_k$ owns a segment $\mathcal{M}:<\mathcal{V}, \mathcal{E}>$ and needs to produce correct output features $\mathcal{F}^k$. Once the $\mathcal{M}$ and $\mathcal{F}^k$ is given, we need to calculate the necessary input feature size for every layer $l_i \in \mathcal{M}$.
The calculation had been discussed in \cite{zhao2018deepthings}, but it only considered models with chain structure. We will extend it into a more complex graph structure here with a top-down algorithm. 

To calculate the input feature size of layer $l_i$, we need to find all the edges $(l_i, l_j)$ start from $l_i$. We can assume the input feature sizes of all $l_j$ is already calculated. Since the input of $l_j$ is just the output of $l_i$, the necessary output feature size of $l_i$ can be denoted as:
\begin{equation}
	{w}_{i} = \max {\{w_{i \to j} \}},
	\ {h}_{i} = \max { \{h_{i \to j}\} }.
\end{equation}
Here we use $w_i$ and $h_i$ to denote the necessary width and height of the output feature size of $l_i$, meanwhile, $w_{i \to j}$ and $h_{i \to j}$ is the input size of layer $l_j$.

Assume layer $l_i$ has $k^w_i \times k^h_i$ kernel size and $s_{i}$ stride size, once the output feature size is determined, the height $h_i$ and width $w_i$ of input feature can be calculated using the following equation:
\begin{equation} \label{eq:necessary-wihi}
	{w}_{* \to i} = ({w}_{i}-1)s_{i} + k_{i}^w, \ {h}_{* \to i} = ({h}_{i}-1)s_{i} + k_{i}^h
\end{equation}
where ${w}_{* \to i}$ and ${h}_{* \to i}$ is the input feature size for $l_i$. Note this formula suits for both conv and pool layers.

Since the output feature size of all sink vertices of $\mathcal{M}$ is given (corresponding to $\mathcal{F}^k$), we can iteratively calculate all the output and input feature size of all layers in $\mathcal{M}$ with a top-down algorithm. 
The input feature size of all the source vertices of $\mathcal{M}$ is the input feature size needed by device $d_k$. 

\subsubsection{Inference Cost Of Devices}
We use $f(l_i; F_i^k)$ to denote the required floating operations (FLOPs) of conv layer $l_i$ when generating an output feature map $F_i^k$ with size $c_i \times w_i \times h_i$.
Assume layer $l_i$ is a conv layer with $ c_i' \times k_i^w \times k_i^h$ kernel size, $c_i$ output channel and $s_i$ stride size. Since each floating scalar in the output feature is calculated by sliding the kernel over the input feature,
$f(l_i; F_i^k)$ can be given by:
\begin{equation}
	f(l_i; F_i^k) = {k_i^w} {k_i^h} c_i' w_i h_i  c_i.
\end{equation}
Here we ignore the pool layers since they require far fewer FLOPs than conv layers (In Fig. \ref{fg:layer-overhead}).

Note the $w_i$ and $h_i$ in Eq. \eqref{eq:necessary-wihi} denote the region of correct output feature. However, the $F_i^k$ is the actual output feature size, which contains not only the correct output but also some redundant parts at the margin of $F_i^k$. 
Assume the size of $F_i^k$ is known and $(l_i, l_j)$ is an edge in $\mathcal{M}$, the output $F_j^k$ of layer $l_j$ can be calculated by:
\begin{equation} \label{eq:to-next-wh}
	w_j = \frac{w_i + 2 p_j^w - k_j^w}{s_j^w} + 1 , \ h_j = \frac{h_i + 2 p_j^h - k_j^h}{s_j^h} + 1
\end{equation}
where $p_j^w$ and $p_j^h$ is the padding size of conv layer $l_j$. Since we have the input feature size for all source vertices in $\mathcal{M}$, to calculate the $F_i^k$ for all layer $l_i \in \mathcal{M}$, we use a bottom-up algorithm similar with above, omit here.

Assume a device $d_k$ is responsible to produce $\mathcal{F}^k$ with model segment $\mathcal{M}$, we can give the required FLOPs operation $\theta(\mathcal{M} ; \mathcal{F}^k)$ with:
\begin{equation}
	\theta(\mathcal{M} ; \mathcal{F}^k) = \sum_{l_j \in \mathcal{M}} {f(l_j; F_j^k)}.
\end{equation}
Empirical studies by \cite{deepslicing} have demonstrated that for specific layers and device, the computation time is proportional to the size of the input or output features, Therefore, the inference time $t_{comp}(d_k, F_k)$ for device $d_k$ can be estimated by the following equation:
\begin{equation}
	t_{comp}(d_k, \mathcal{F}^k)) = \alpha_k \frac{\theta(\mathcal{M} ; \mathcal{F}^k)}
	{\vartheta(d_k)}
\end{equation}
where 
$\vartheta(d_k)$ is the computing capacity (FLOPS) of device $d_k$. $\alpha_k$ is a coefficient computed by a regression model.

\subsubsection{The Period and Latency Of Pipeline}
As each device executes inference in parallel within stage, the computation time for stage $\mathcal{S}: <\mathcal{M}, \mathcal{D}>$ is determined by the maximum inference time among devices in $\mathcal{D}$:
\begin{equation}
	T_{comp}(\mathcal{S}) 
	= \max_{d_k \in \mathcal{D}}
	{t_{comp}(d_k, \mathcal{F}^k)}.
\end{equation}

Since each device $d_k \in \mathcal{D}$ will generate part of the calculation of stage $\mathcal{S}$, there exists a device $d_f$  which is responsible to distribute stage input and gather stage output from other devices. For a device $d_k \in \mathcal{S}$, the feature transferring time $t_{comm}( d_f, d_k, \mathcal{M})$ can be given by:
\begin{equation}
	t_{comm}( d_f, d_k, \mathcal{F}) = \frac{\varphi(\mathcal{F}^k_{in}) + \varphi(\mathcal{F}^k_{out})}{b(d_f, d_k)}
\end{equation}
where $\varphi(\mathcal{F})$ is the feature size on a given input feature sizes $\mathcal{F}$. Here we use $\mathcal{F}^k_{in}$ and $\mathcal{F}^k_{out}$ to denote the input and output feature sizes of $\mathcal{M}$ owned by $d_k$.
Sum the communication cost for each device $d_k$ in stage $\mathcal{S}$, we define
\begin{equation}
	T_{comm}(\mathcal{S}) 
	= \sum_{\substack{d_k \in \mathcal{D} \\ d_k \neq d_f }}
	{t_{comm}(d_f, d_k, \mathcal{F}^k)}
\end{equation}
as the communication cost of stage $\mathcal{S}$.

The cost function for each stage in pipeline inference is then defined as the total time of the frame transfer and layer computation: 
\begin{equation}
	T({\mathcal{S}}) = T_{comp}(\mathcal{S}) + T_{comm}(\mathcal{S}) \label{eq:stage-time}
\end{equation}
Note the time for feature map partition and stitch is not discussed here. In practice, it is far less than the layer computation time $T_{comm}(\mathcal{S})$ and could be ignored.

Next, we define the optimization objective as:
\begin{equation}
	\mathcal{P}(\mathbb{G}, \mathbb{D}, \mathbb{S})=\max_{\mathcal{S} \in \mathbb{S}}
	T({\mathcal{S}}), \ 
	\mathcal{T}(\mathbb{G}, \mathbb{D}, \mathbb{S})=\sum_{\mathcal{S} \in \mathbb{S}}
	T({\mathcal{S}})
\end{equation}
where $\mathcal{P}(\mathbb{G}, \mathbb{D}, \mathbb{S})$, $\mathcal{T}(\mathbb{G}, \mathbb{D}, \mathbb{S})$ estimate the maximum execution time of stages in and inference latency in pipeline.

\subsection{Analysis} \label{sec:np-analysis}
\begin{table}
	\centering
	\caption{Optimization Complexity} \label{tab:optim-complex}
	\begin{threeparttable}[tb]
		\centering
		\begin{tabular}{P{0.25\linewidth}
				P{0.15\linewidth}
				P{0.15\linewidth}
				P{0.15\linewidth}
			} 
			\toprule
			\diagbox{\textbf{Device}}{\textbf{Model}} & \textbf{Chain} & \textbf{Block} & \textbf{Graph} \\ \midrule
			\textbf{Homogeneous}                      & P              & NP \tnote{*}   & NP             \\
			\textbf{Heterogeneous}                    & NP             & NP             & NP             \\ \bottomrule
		\end{tabular}
		\begin{tablenotes}
			\footnotesize
			\item[*] \cite{zhou2019adaptive} solves the optimization by considering the entire block as a special layer. However, this operation introduces lots of unnecessary calculations during inference.
		\end{tablenotes}
	\end{threeparttable}
\end{table}

The goal of our optimization algorithm is finding the best stage set $\mathbb{S}^\star$ that minimizes the maximum period $\mathcal{P}(\mathbb{G}, \mathbb{D}, \mathbb{S})$ of pipeline with heterogeneous clusters. Such an optimization faces the following challenges:
\begin{itemize}
	\item The overhead of computation and communication of each layer in model varies and would be affected by the assigned feature map size $F_i^k$.
	\item The computing capacity $\vartheta(d_k)$ of every device in the heterogeneous cluster varies.
	\item For a specific stage $\mathcal{S}$, the number of devices $|\mathcal{D}|$ and the model segment $\mathcal{M}$ in stage also need to be configured.
	      
	\item The structure of CNN model can be complex and hard to be partitioned.
\end{itemize}

In fact, we show the optimal solution can not be found in polynomial time unless $P=NP$.

\begin{theorem} \label{th:np-hard-chain}
	Given a CNN model $\mathbb{G}$ with chain structure, the problem of minimizing maximum stage execution time $\mathcal{P}(\mathbb{G}, \mathbb{D}, \mathbb{S})$ with heterogeneous mobile devices under a constriction that $\mathcal{T}(\mathbb{G}, \mathbb{D}, \mathbb{S}) \leq T_{lim}$ is NP-hard.
\end{theorem}

\begin{proof}
	Considering a scheduling problem defined as follows: Given $L$ identical tasks that are needed to be executed one by one. All tasks can be paralleled to several processors without additional overhead. The goal is to assign these tasks to $D$ heterogeneous devices and maximize the throughput. This problem is proven to be NP-hard by \cite{benoit2010complexity}. 
	We can construct a CNN model with chain structure whose layers are identical and the kernel size of each layer is $1 \times 1$. This kernel size guarantees there is no overlapped partition when parallels the inference. If there exists a polynomial solution for this CNN model, obviously it can also be applied to the above task assignment problem. Thus, the optimization of $\mathcal{P}(\mathbb{G}, \mathbb{D}, \mathbb{S})$ is NP-hard. Here complete the proof.
\end{proof}

\begin{theorem} \label{th:np-hard-graph}
	Given a CNN model $\mathbb{G}$ with graph structure, the problem of minimizing maximum stage execution time $\mathcal{P}(\mathbb{G}, \mathbb{D}, \mathbb{S})$ with homogeneous mobile devices under a constriction that $\mathcal{T}(\mathbb{G}, \mathbb{D}, \mathbb{S}) \leq T_{lim}$ is NP-hard.
\end{theorem}
\begin{proof}
	We begin with introducing the problem of most balanced $st$-edge cuts (MBSTC). Given a graph $G$, founding an edge cut $[G, \bar{G}]$, which minimize $\max \left\{|G|,|\bar{G}|\right\} $ is NP-Hard \cite{stcuts2010}. Obviously, MBSTC is a special case of our optimization when the number of stages is $2$. Thus, the problem is NP-Hard.
\end{proof}

Given a heterogeneous edge environment, Theorem \ref{th:np-hard-chain} shows find the optimal solution for CNNs with chain structure is NP-Hard. Theorem \ref{th:np-hard-graph} shows that for CNNs with graph structure, even homogeneous edge environment is NP-Hard.

We summarize the result in Table \ref{tab:optim-complex}, almost every situation is NP-Hard for optimization except chain structure model with homogeneous devices.

\section{Orchestrate The Model Structure} \label{sec:graph}
In this section, we introduce our strategy to orchestrate the complex block and graph structures.

\subsection{Insight}

Ideally, we hope to directly divide CNN model $\mathbb{G}$ and mobile devices $\mathbb{D}$ into several stages $\mathbb{S}$. 
However, we can find there is no polynomial solution for $\mathbb{G}$ with block and graph structures from Table \ref{tab:optim-complex}, neither with homogeneous nor heterogeneous environment. The only feasible situation to find the optimal strategy $\mathbb{S}$ is that the structure of model is a chain.

For block structure, a simple trade-off is to consider every block as a special layer. So that it could be optimized in polynomial time. However, this scheme introduces lots of redundant calculation inside blocks and can not be applied on these models with graph structures.
Fig. \ref{fg:graph-stacks} shows an extreme case of such a scheme. Considering a block with only two conv layers ($l_a$, $l_b$). The kernel size of $l_a$ is $1 \times 7$ but the kernel size of $l_b$ is $7 \times 1$. 
For layer $l_a$, according to Eq. \eqref{eq:necessary-wihi} there is no redundant calculation on its width dimension since $k^w_a = 1$ (assuming the stride size is $1$). 
Similarly, there is no redundant calculation on its length dimension for layer $l_b$. However, if we regard the block as a special whole layer, it will have redundant calculation on both width and length, as shown in Fig. \ref{fg:arrange-before}.
But the block can be divided into two sequential \textit{pieces}, one piece contains the input vertex and layer $l_a$, the other piece contains layer $l_b$ and output vertex as shown in Fig. \ref{fg:arrange-after}. After this operation, there is no more redundant calculations inside the two pieces.

Here comes the insight. Given a CNN model $\mathbb{G}$, the goal is to transform it into a sequence of pieces. Since there may be some redundant calculation inside pieces, we need to minimize the redundancy inside every piece.
After this operation, each piece can be regarded as a layer of original $\mathbb{G}$. Since these pieces construct a chain structure, the operation gives change for further optimization.

\begin{figure}[tb]
	\centering
	\subfloat[]{ \label{fg:arrange-before}
		\includegraphics[width=0.45\linewidth]{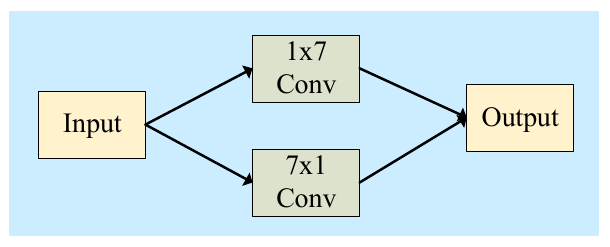}
	}
	\subfloat[]{ \label{fg:arrange-after}
		\includegraphics[width=0.52\linewidth]{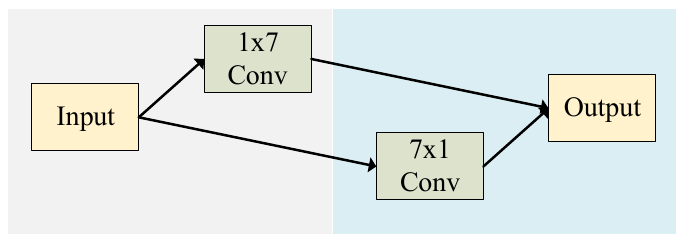}
	}
	\caption{A extreme case for a block with two layers. (a): consider the entire block as a special layer. (b): partition the block into more fine-grained pieces.}
	\label{fg:graph-stacks}
\end{figure}

\begin{figure}[tb]
	\centering
	\subfloat[]{ \label{fg:ending-stacks-a}
		\includegraphics[width=0.23\linewidth]{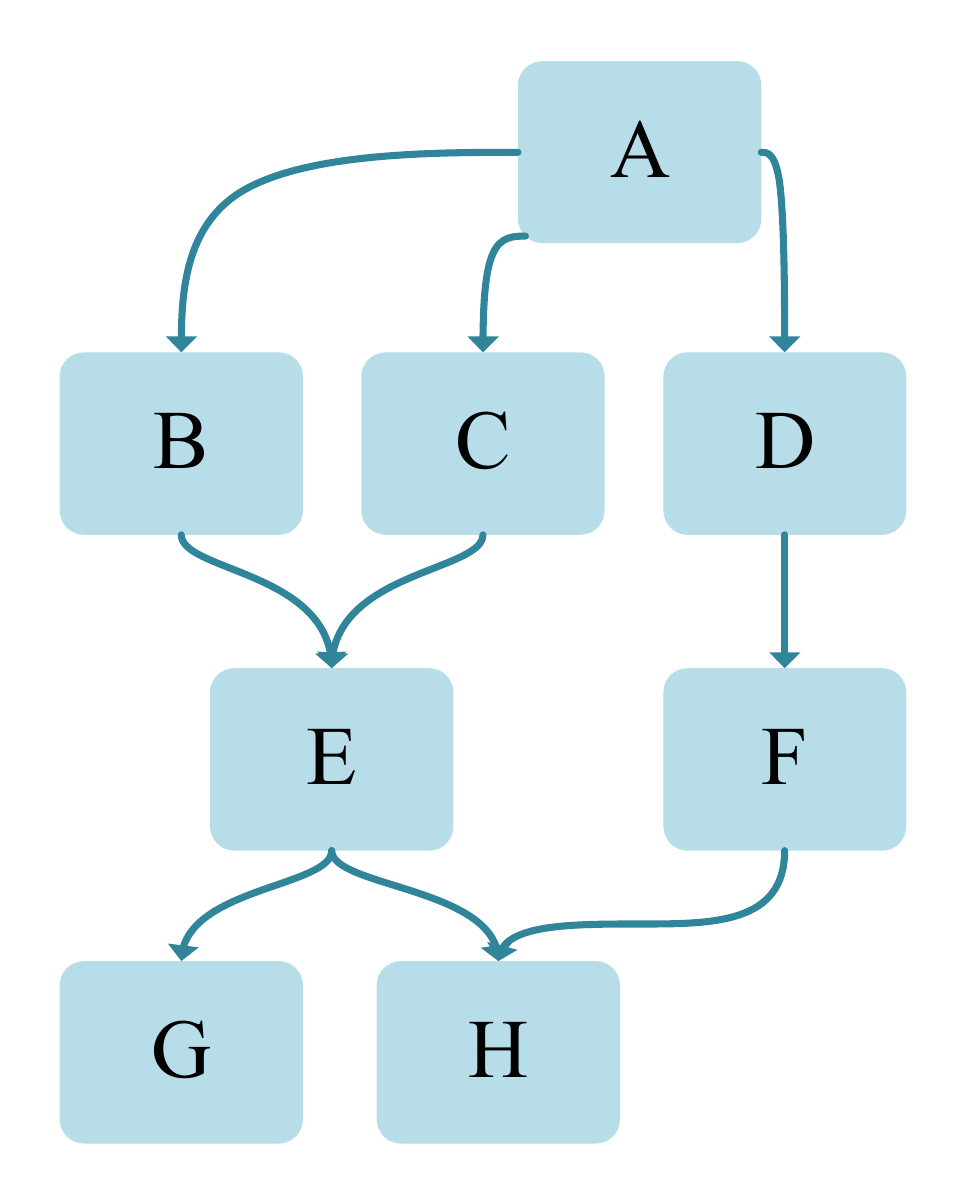}
	}
	\subfloat[]{ \label{fg:ending-stacks-b}
		\includegraphics[width=0.23\linewidth]{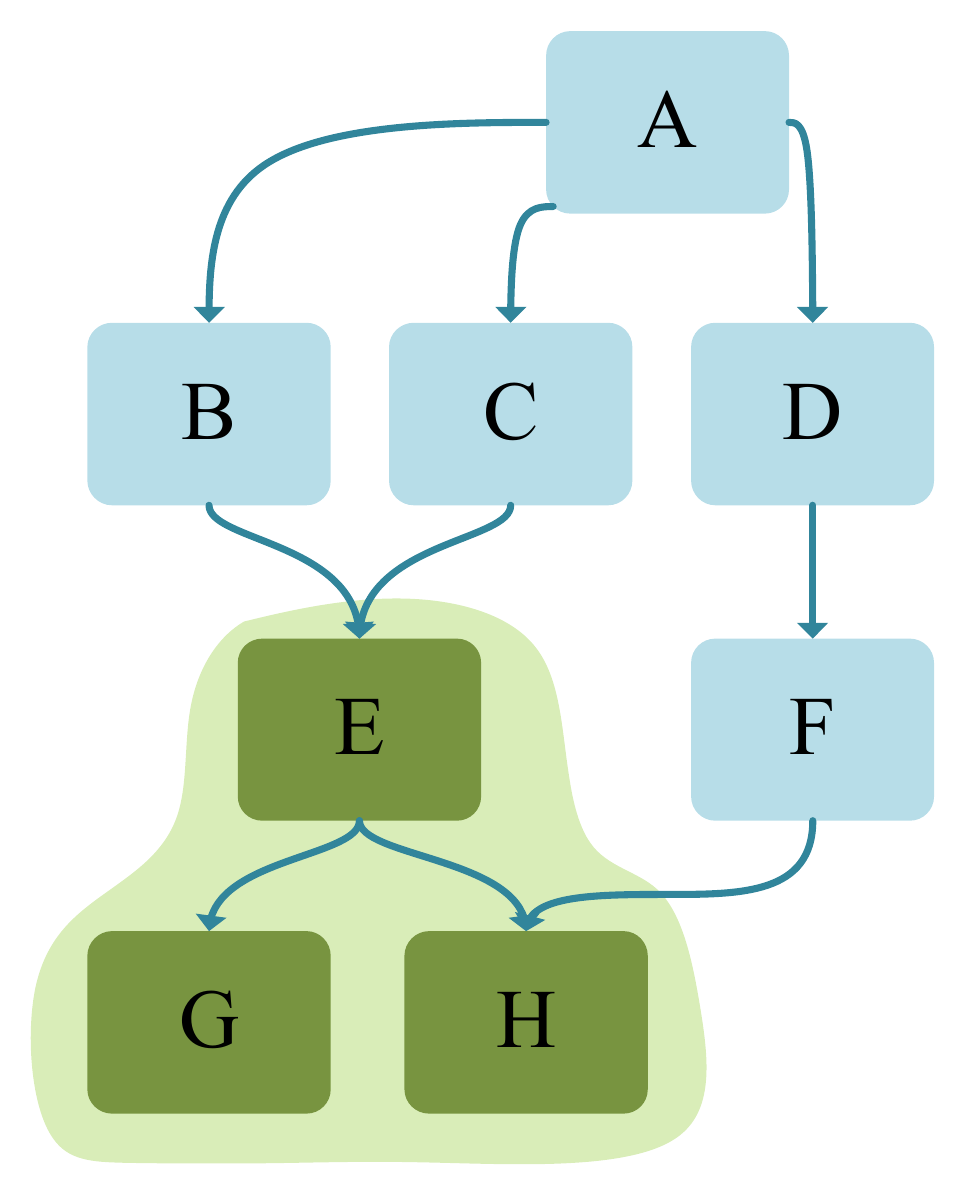}
	}
	\subfloat[]{ \label{fg:ending-stacks-c}
		\includegraphics[width=0.23\linewidth]{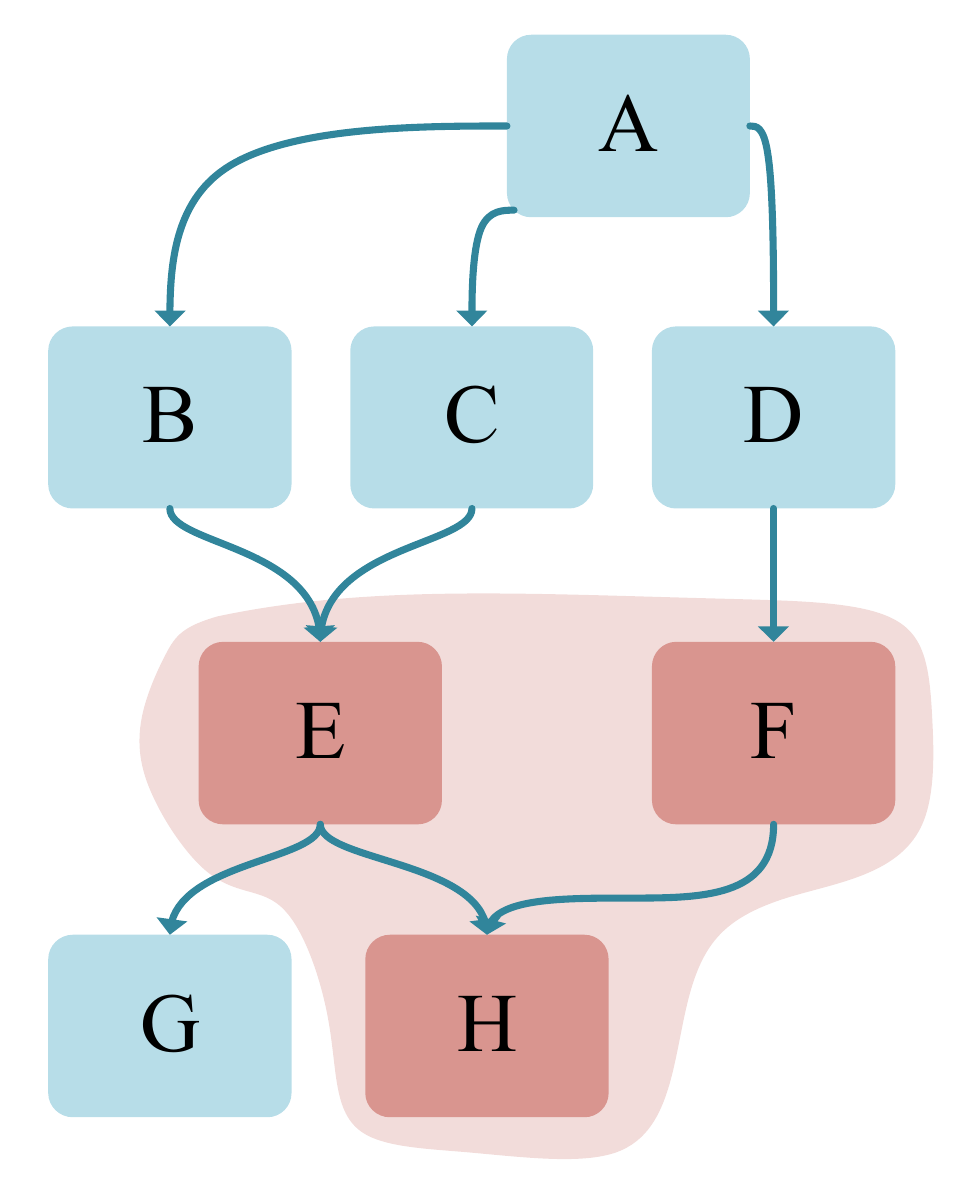}
	}
	\subfloat[]{ \label{fg:ending-stacks-d}
		\includegraphics[width=0.23\linewidth]{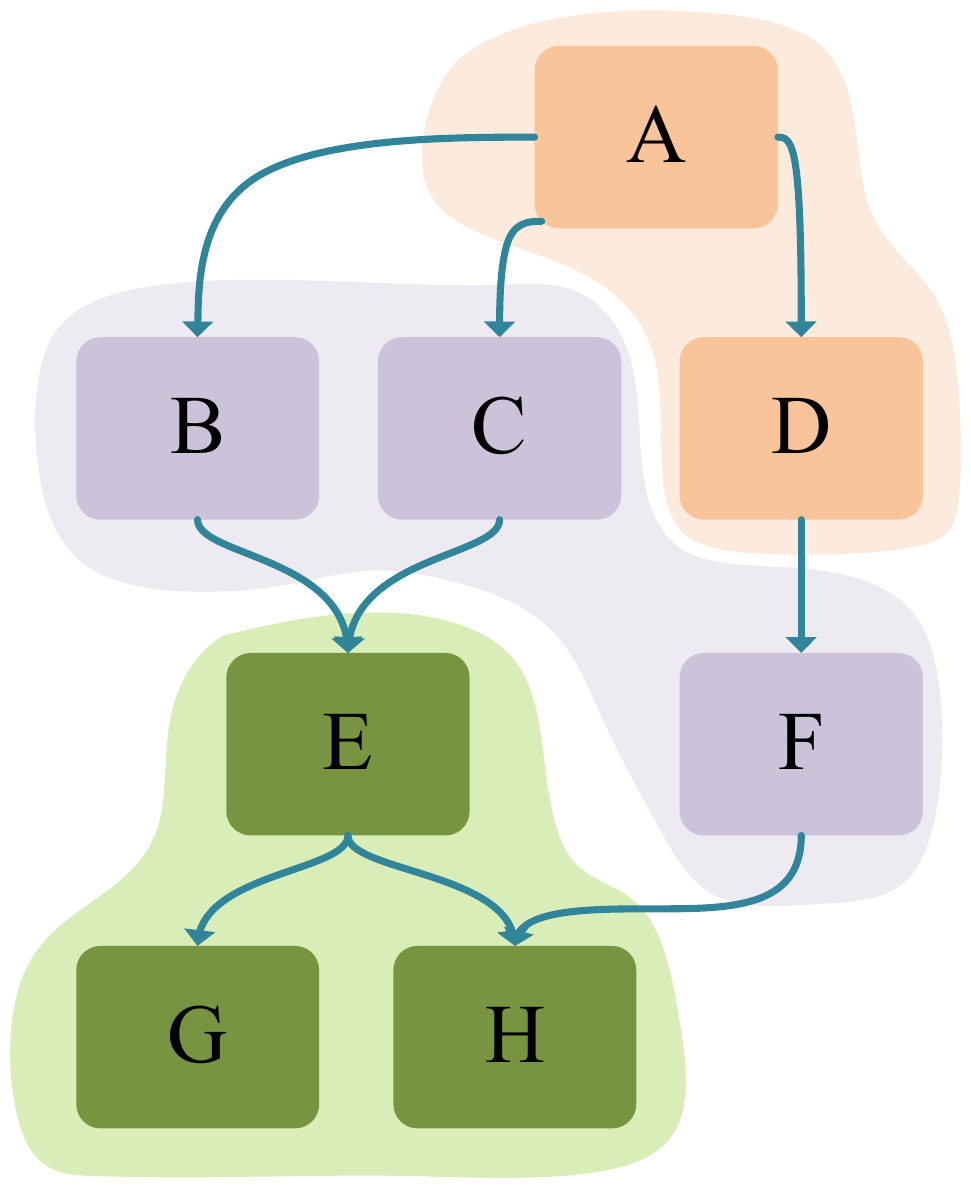}
	}
	\caption{The illustration of ending pieces. (a): The original graph $\mathbb{G}$. (b): $\{E,G,H\}$ is an ending piece of $\mathbb{G}$. (c): $\{E,F,H\}$ is not an ending piece of $\mathbb{G}$. (d): A partition of $\mathbb{G}$ using ending piece iteratively. }
	\label{fg:ending-stacks}
\end{figure}

\subsection{Partition Graph Into Pieces}

We give a graph partition algorithm based on dynamic programming. 

Here we will reveal the existence of the optimal substructure property of the problem of partitioning graph into pieces, which is necessary for dynamic programming. We first define the concept of \textit{ending} piece of graph $\mathbb{G}$. Note since the piece of $\mathbb{G}$ is just a \textit{smaller segment} defined in the previous section, we still use the notation $\mathcal{M}$ to denote these pieces.
\begin{definition}
	An \textit{ending piece} $\mathcal{M}_E$ is a special piece of $\mathbb{G}$ which for any edge $(u, v) \in \mathbb{G}$, if $u \in \mathcal{M}_E$, then $v \in \mathcal{M}_E$.
\end{definition}

Fig. \ref{fg:ending-stacks} gives an illustration of ending pieces. We plot a small graph $\mathbb{G}$ with $8$ vertices in Fig. \ref{fg:ending-stacks-a} and two different pieces in Fig. \ref{fg:ending-stacks-b} and Fig. \ref{fg:ending-stacks-c}. 
The green region $\{E, G, H\}$ in Fig. \ref{fg:ending-stacks-b} is an ending piece of original graph $\mathbb{G}$. But the red region $\{E,F,H\}$ in Fig. \ref{fg:ending-stacks-c} is not an ending piece since the edge $E$ is a member of this piece but $G$ is outside the red region. 
Graph $\mathbb{G}$ can be partitioned into pieces using the concept of ending piece recursively. Given a graph $\mathbb{G}$, the sketch of the procedure is to find an ending piece $\mathcal{M}_E$ and add it to the partition result (a sequence of pieces), then consider $\mathbb{G} - \mathcal{M}_E$ as a new graph and repeat the previous procedure recursively. Fig. \ref{fg:ending-stacks-d} shows a partition result of graph $\mathbb{G}$.

Note partition $\mathbb{G}$ by ending piece can not guarantee that these pieces obtained are with a chain structure. If we assign vertex $B$ in Fig. \ref{fg:ending-stacks-d} from the middle piece $\{B, C, F\}$ to the first piece $\{A, D\}$, the partition result can also be obtained by the above procedure. 
However, such a result do not satisfy our goal since the first piece connects with two pieces at the same time. 
To prevent this, we add a constraint that all the vertices that are directly connect to the ending piece must belong to the ending piece in the next iteration.
With this constraint, once $\{E,G,H\}$ is determined as ending piece, $\{B, C, F\}$ must belong to the same piece in the final result, which guarantees the obtained pieces a chain structure.

\subsection{The Algorithm}
Since our goal to minimize the redundant calculation of every piece, we need to quantify the redundant calculation cost $C(\mathcal{M})$ for a given piece $\mathcal{M}$.
Assume $\mathcal{I}$ is the original input feature sizes of sources nodes in $\mathcal{M}$, and $\mathcal{I}'$ is the feature size with redundant parts that are calculated with Eq. \eqref{eq:necessary-wihi}. The value of $C(\mathcal{M})$ can be easily quantified by the difference of required FLOPs for the two input. 

Here we give the state transfer equation for partitioning the graph structure:
\begin{equation} \label{eq:dp-piece}
	F(\mathbb{G}) = \min_{\mathcal{M}_E \subset \mathbb{G}} \left\{ \max 
	\left\{ F(\mathbb{G}-\mathcal{M}_E), C(\mathcal{M}_E) \right\}
	\right\}.
\end{equation}
If $\mathbb{G}$ is partitioned into multiple pieces $\mathcal{M}$, the function $F(\mathbb{G})$ return the minimum FLOPs difference $C(\mathcal{M})$ among all partitioned pieces in $\mathbb{G}$. Algorithm \ref{al:graph2stack} gives the pseudocode.

\begin{algorithm}[tb]
	\caption{Partition graph into pieces}
	\small
	\label{al:graph2stack}
	\begin{algorithmic}[1]
		\Require $F$: map indexed with the hash $h(\mathbb{G})$ and return $F(\mathbb{G})$.
		\Require $R$: map indexed with the hash $h(\mathbb{G})$ and return $\mathcal{M}$.
		\Function{partition}{$\mathbb{G}$, $\mathcal{M}_E'$}
		\State compute the hash $h(\mathbb{G})$
		\If{$F$ contains $h(\mathbb{G})$}
		
		\Return $F[h(\mathbb{G})]$
		\EndIf
		\State $min \gets \infty$
		
		\State get $\mathcal{M} \subset h(\mathbb{G})$ which directly is connected with $\mathcal{M}_E'$
		\For{$\mathcal{M}_E \gets DFS(\mathbb{G} - \mathcal{M})$}
		\State $\mathcal{M}_E \gets \mathcal{M}_E \cup \mathcal{M}$
		\State calculate the redundancy $C(\mathcal{M}_E)$
		\State $cur \gets \max(\text{partition}(\mathbb{G} - \mathcal{M}_E$, $\mathcal{M}_E), C(\mathcal{M}_E))$
		\If{$min > cur$}
		\State $F[h(\mathbb{G})] = cur$
		\State $R[h(\mathbb{G})] = \mathcal{M}_E$
		\State $min \gets cur$
		\EndIf
		\EndFor
		
		\Return $F[h(\mathbb{G})]$
		\EndFunction
		
		\Function{obtain}{$\mathbb{G}$}
		\If{$\mathbb{G}$ == $\phi$}
		
		\Return
		\EndIf
		\State $\mathcal{M} \gets R[h(\mathbb{G})]$
		\State print the piece $\mathcal{M}$
		\State obtain($\mathbb{G} - \mathcal{M}$)
		\EndFunction
	\end{algorithmic}
\end{algorithm}

\textbf{Line 2-4}
checks whether $F(\mathbb{G})$ is already calculated, if true, the following computation can be skipped. Otherwise, we use a variable $min$ to store the minimum value located in Eq. \eqref{eq:dp-piece}.

\textbf{Line 5-7} adds a constraint to make sure the partitioned pieces follow a chain structure. Since the \textit{partition} function uses recursion, the parameter $\mathcal{M}_E'$ stores the partitioned piece in its previous calculation. We use a \textit{DFS} function to produce all the possible $\mathcal{M}_E$.

\textbf{Line 8-13} is the core part of our proposed dynamic programming. It iterates all possible $\mathcal{M}_E$, and uses recursion to solve the optimization problem. Here we use a variable $cur$ to store the best partition strategy for current $\mathcal{M}_E$. If the current strategy is better than the one we have recorded, we update the record variable $min$ and map $F$ and $R$. 

\textbf{Line 14-18} is the \textit{obtain} function that receives the CNN $\mathbb{G}$ and shows the best partition strategy using the map $R$ that is calculated in the \textit{partition} function.

Note the \textit{DFS} function can produce tons of available $\mathcal{M}_E$ for a given $\mathbb{G}$. Since iterating all of them leads to unfeasible complexity for optimization, we use a simple pruning strategy here. From the above discussion, it is clear that the more sequential layers we fuse, the more redundancy we get. In fact, we observe that the redundancy is intolerable when the \textit{diameter} of $\mathcal{M}_E$ exceeds a specific number.
\begin{definition}
	The diameter of piece $\mathcal{M}$ is the greatest distance of any vertex pair in $\mathcal{M}$.
\end{definition}
With this observation, we limit the diameter of produced $\mathcal{M}_E$ in \textit{DFS} function under a constant integer $d$. In practice, we set the value of $d$ to $5$. \label{sec:dia-desc}

\section{Pipeline Cooperation for CNN Inference} \label{sec:main-algo}

\begin{algorithm}[tb]
	\caption{Dynamic programming for pipeline inference}
	\small
	\label{al:dp-homo}
	\begin{algorithmic}[1]

		\Require $P, L$: 3D arrays used to record the period and latency.
		\Require $S, R$: 3D array used to trace the computed stage and sub-pipeline.
		
		\Function{DP}{$i$, $j$, $p$, $T_{lim}$}
		\If {$P[i][j][p]$ exists}
		
		\Return $P[i][j][p]$, $L[i][j][p]$
		\EndIf
		
		\State calculate $Ts[i][j][p]$ using \eqref{eq:stage-time}
		\State $P[i][j][p] \gets Ts[i][j][p]$
		\State $T[i][j][p] \gets Ts[i][j][p]$
		\State $S[i][j][p] \gets (i, j, p)$
		
		\If {$m = 1$ or  $j=i+1$}
		
		\Return $P[i][j][p]$, $T[i][j][p]$
		\EndIf
		
		\For{$s := i \to j-1$}
		\For{$m := 1 \to p - 1$}
		\State calculate $Ts[s+1][j][m]$ using \eqref{eq:stage-time}
		\State $T_{lim} \gets T_{lim} - Ts[s+1][j][m]$
		\If {$T_{lim} < 0$}
		\Continue
		\EndIf
		\State $P[i][s][p-m], T[i][s][p-m] \gets \text{DP}(i, s, p-m, T_{lim})$
		\If {$T_{lim} < T[i][j][p-m]$}
		\Continue
		\EndIf
		\State $ period \gets \max(P[i][s][p - m], Ts[s+1][j][m])$
		\If {$ period < P[i][j][p])$}
		\State $P[i][j][p] \gets period$
		\State $T[i][j][p] \gets T[i][s][p-m] + Ts[s+1][j][m]$
		\State $R[i][j][p] \gets (i, s, p-m)$
		\State $S[i][j][p] \gets (s+1, j, m)$
		\EndIf
		\EndFor
		\EndFor
		\Return $P[i][j][p]$, $L[i][j][p]$
		\EndFunction
		
		\Function{BuildStrategy}{($i$, $j$, $p$), $\mathbb{S}$}
		\If {$R[i][j][p]$ }
		\State BuildStrategy($R[i][j][p]$, $\mathbb{S}$)
		\EndIf
		\State calculate $ \mathcal{S}_{i \to j}$ using $S[i][j][p]$
		
		\State $\mathbb{S} \gets \mathcal{S}_{i \to j} \cup \mathbb{S}$
		\EndFunction
	\end{algorithmic}
\end{algorithm}

In this section, we present a pipeline cooperation (PICO) scheme aimed at efficiently executing CNN inference. PICO uses a heuristic algorithm based on dynamic programming to optimize the inference pipeline. We also implement an adaptive framework that automatically chooses suitable parallel scheme under dynamic workload. 

\subsection{Heuristic}

For chain structure, although the polynomial algorithm for $\mathcal{P}(\mathbb{G}, \mathbb{D}, \mathbb{S})$ does not exist unless $P=NP$, the optimal solution can be found in polynomial time if these mobile devices are homogeneous, which leads to a heuristic two-step algorithm. 
We first find the optimal ${\mathbb{S}}^\star$ for a homogeneous cluster, then adapt the $\mathbb{S}^\star$ to a heterogeneous cluster using a greedy algorithm.

Since the CNN $\mathbb{G}$ is partitioned into $L$ pieces, considering a specific stage $\mathcal{S}: <\mathcal{M}, \mathcal{D}, \mathcal{F}>$ that starts from $i$-th piece and ends at $j$-th piece. We can use the notation $\mathcal{S}_{i \to j}$ to represent it, so to the two notations $\mathcal{M}_{i \to j}$ and $\mathcal{D}_{i \to j}$.
\subsubsection{Dynamic Programming}

Based on the given cluster $\mathbb{D}$, we construct a new cluster $\mathbb{D}'$, which has the same number of devices of $\mathbb{D}$, but the computing capacity of each device is equivalent to the average of  $\mathbb{D}$.
\begin{equation}
	\vartheta({d'}_k) = \frac{\sum_{d_k \in \mathbb{D} } \vartheta(d_k)}{|\mathbb{D}|}\ \forall {d'}_k \in \mathbb{D}',
	\quad |\mathbb{D}'| = |\mathbb{D}|
\end{equation}

For any device $d_k$ belongs to this stage, the output feature size $\mathcal{F}^k$ is equivalently partitioned. Thus, $\mathcal{F}^{k}$ can be determined by the size of stage. We denote $p = |\mathcal{D}_{i \to j}|$ for convenience.

The expression of stage can now be simplified as a three-element tuple $(i, j, p)$.
For the optimal pipeline $\mathbb{S}^\star$, it can now be broken into an optimal sub-pipeline consisting of pieces form $1$ through $s$ with $p-m$ mobile devices followed by a single stage with pieces $s+1$ through $j$ replicated over $m$ workers. Then using the optimal sub-problem property, we can solve the optimization problem through dynamic programming:
\begin{equation} \label{eq:beq2}
	P[i][j][p] \! = \! \min_{i \leq s < j} \min_{1 \leq m < p} \max
	\begin{cases}
		P[i][s][p-m] \\
		Ts[s+1][j][m]
	\end{cases}
\end{equation}
where $P[i][s][p-m]$ is the time taken by the slowest stage of the optimal sub-pipeline between piece $i$ and $s$ with $p-m$ edge devices, $Ts[s+1][j][m]$ is the time taken for a stage with model segment $\mathcal{M}_{s+1 \to j}$ with $m$ devices. Obviously $P[1][L][D]$ is equivalent to $\mathcal{P}(\mathbb{G}, \mathbb{D}', \mathbb{S})$ in the homogeneous case. During optimization, we prune these solutions that exceed the inference limitation $T_{lim}$.

Algorithm \ref{al:dp-homo} shows the pseudocode of our optimization algorithm which uses dynamic programming with memorization to find out the optimal parallelization strategy. Function \textit{DP} computes the minimum period and records the optimal pipeline configuration in two 3D arrays $R$ and $S$. 
The optimal parallelization strategy is built up through function \textit{BuildStrategy} by recursively iterating the calculated $R$ and $S$, and adding the corresponding stage configuration $\mathcal{S}_{i \to j}$ to $\mathbb{S}$.

\subsubsection{Adapt to the heterogeneity}

We use a greedy algorithm to adapt the calculated $\mathbb{S}'$ in Algorithm \ref{al:dp-homo} to the heterogeneous environment. For every stage $\mathcal{S}_{i \to j}' \in \mathbb{S}'$, we keep the model segment $\mathcal{M}_{i \to j}$ unchanged and choose a proper set of edge devices as $\mathcal{D}_{i \to j}$ from heterogeneous cluster $\mathbb{D}$. Let $\Theta_{i \to j}$ and $\Theta_{i \to j}'$ denotes the required computing resources of stage $\mathcal{S}_{i \to j}$ and $\mathcal{S}_{i \to j}'$:
\begin{equation}
	\Theta_{i \to j} = \sum_{d_k \in \mathcal{D}_{i \to j}} \theta(\mathcal{M}_{i \to j}; \mathcal{F}^k),
\end{equation}
We want $\Theta_{i \to j}$ to be as close to $\Theta_{i \to j}'$ as possible.

We initialize the stage set $\mathbb{S}$ with the same number of stages, each stage only the same number of workers and the same model fragment $\mathcal{S}_{i \to j}$. 
To achieve our goal, we sort the mobile devices by the computing capabilities $\vartheta(d_k)$ in reverse order and then iterate each device. In every iteration, we find the stage $\mathcal{S}_{i \to j}' \in \mathbb{S}'$ with maximum average computing requirement $\frac{\Theta_{i \to j}'}{|\mathcal{D}_{i \to j}'|}$.
The current device $d_k$ will be added to device set $\mathcal{D}_{i \to j}$. Once $\mathcal{D}_{i \to j}$ owns the same number of device in $\mathcal{D}_{i \to j}'$, we adjust the output feature size $\mathcal{F}^k$ for every device $d_k \in \mathcal{D}_{i \to j}$ with a \textit{Divide And Conquer} algorithm. After this operation, we accomplish the presentation of stage $\mathcal{S}_{i \to j}$ and add it to $\mathbb{S}$. After all the iterations, we have a set of stages  $\mathbb{S}$ for the heterogeneous cluster. The complete algorithm is shown in Algorithm \ref{al:greedy-hete}. 

\begin{algorithm}[tb]
	\caption{Adjust stage configuration $\mathbb{S}$ for heterogeneous devices}
	\small
	\label{al:greedy-hete}
	\begin{algorithmic}[1]
		\Require $\mathbb{S}'$: the optimal stage set for homogeneous cluster.
		
		\algnewcommand\And{\textbf{and}}
		\Function{AdjustStage}{}
		
		\State Initialize an empty $\mathbb{S}$
		\State Sort devices in $\mathbb{D}$ by computing capabilities $\vartheta(d_k)$
		\For{$d_k \in \mathbb{D}$}
		\State Find the stage $\mathcal{S}_{i \to j}' \in \mathbb{S}'$ with minimum $\frac{\Theta_{i \to j}'}{|\mathcal{D}_{i \to j}'|}$
		\State Get $\mathcal{S}_{i \to j}$ from $\mathbb{S}$ or create $\mathcal{S}_{i \to j}$ with empty $\mathcal{D}_{i \to j} $
		\State $\mathcal{D}_{i \to j} \gets d_k \cup \mathcal{D}_{i \to j}$
		\State Remove one device from $\mathcal{D}_{i \to j}'$
		
		\If {$|\mathcal{D}_{i \to j}'| = 0$}
		\State Adjust feature partition $\mathcal{F}^k$ for every $d_k \in \mathcal{D}_{i \to j}$.
		\State $\mathbb{S} \gets \mathcal{S}_{i \to j} \cup \mathbb{S}$
		\State Remove $\mathcal{S}_{i \to j}'$ from $\mathbb{S}'$
		\EndIf
		\EndFor
		\Return $\mathbb{S}$
		\EndFunction
		
	\end{algorithmic}
\end{algorithm}

\subsubsection{Correctness}
PICO can not guarantee the final configuration for the inference pipeline is optimal since the problem is NP-Hard.
But we could show that both Algorithm \ref{al:graph2stack} and \ref{al:dp-homo} find the optimal solutions for the sub-problem they focus on.

\begin{theorem}
	Given a CNN model $\mathbb{G}$, $F(\mathbb{G})$ in Eq. \eqref{eq:dp-piece} returns the maximum amount of redundant FLOPs among those pieces in the optimal arrangement of $\mathbb{G}$.
\end{theorem}

\begin{theorem}
	Assuming $i$ and $j$ is the start and end index of pieces that need to be deployed, $P[i][j][p]$ in Eq. \eqref{eq:beq2} returns the minimal period of all possible pipeline configurations for $p$ mobile devices.
\end{theorem}
The detailed proof can be found in the appendices.

\subsection{Optimization Complexity} \label{sec:complexity-analysis}
PICO aims to find a many-to-many mapping for various CNNs and heterogeneous mobile devices, which has been shown as NP-Hard in Section \ref{sec:np-analysis}. Here we analyze the complexity of these algorithms proposed by PICO.

\subsubsection{Analysis}
PICO contains three novel algorithms, we analyze them one by one.

\textbf{Algorithm \ref{al:graph2stack}} arranges $\mathbb{G}$ into sequential pieces. We first define the width of CNN to formulate the complexity of Algorithm 1:
\begin{definition}[Width $w$ of CNN] \label{def:width}
	Given a CNN graph $\mathbb{G}$, $w$ is the \textit{width} of $\mathbb{G}$ if there are at most $w$ neural layers in $\mathbb{G}$ such that there is no path connecting any two of them.
\end{definition}

Since for every $\mathcal{M}$ generated in Line 6, Algorithm \ref{al:graph2stack}, the upper bound for every path in $\mathcal{M}$ is $d$ since we limit the length of every path in $\mathcal{M}$ (see Section \ref{sec:dia-desc}), the time complexity of Algorithm \ref{al:graph2stack} can be given by Theorem \ref{th:algo1-com}.
\begin{theorem}[Complexity of Algorithm 1] \label{th:algo1-com}
	The time complexity of PICO is $O(wd(\frac{nd}{w})^w)$, where $d$ is the upper bound for every path in $\mathcal{M}$, $n$ is the number of vertices in $\mathbb{G}$, and $w$ is the width of $\mathbb{G}$.
\end{theorem}

\begin{proof}
	We only give a proof sketch here for ease of reading. The detailed proof is in the appendices. 
	Line 8 in Algorithm \ref{al:graph2stack} dominates the computation and the complexity for computing $C(\mathcal{M}_E)$ for every $\mathcal{M}_E$ is $O(wd)$, since there are up to $wd$ vertices in $\mathcal{M}_E$.
	We need to count all the possible pairs $(\mathcal{G}, \mathcal{M}_E)$ during execution to analyze the complexity of the entire algorithm.
	Any directed acyclic graph $\mathbb{G}$ is also equivalent to a partially ordered set. By applying Dilworth's Theorem on $\mathbb{G}$, we can decompose $\mathbb{G}$ into $w$ disjoint chains $\{\mathbb{V}_i\}$. Thus, any possible pair $(\mathbb{G}, \mathcal{M}_E)$ can be decomposed into a combination between $\mathcal{M}_E$ and these chains.
	Since $\mathcal{M}_E$ is an ending piece of $\mathbb{G}$, $\mathcal{M}_E \cap \mathbb{V}_i$ must be a suffix of $\mathbb{V}_i$. Therefore, the upper bound for the possible pair $(\mathbb{G}, \mathcal{M}_E)$ is $\Pi_{i \in \{1, \cdots, w\}} V_i d$, where $V_i$ is the length of chain $\mathbb{V}_i$. And the maximum of $\Pi_{i \in \{1 \cdots, w\}} V_i d$ achieves when $V_i = V_j = n / w$. Thus, the number of all possible pairs $(\mathbb{G}, \mathcal{M}_E)$ is less than $(\frac{nd}{w})^w$, and the complexity of Algorithm \ref{al:graph2stack} is $O(wd(\frac{nd}{w})^w)$.
	
\end{proof}

\textbf{Algorithm \ref{al:dp-homo}} generates an inference pipeline for $D$ homogeneous mobile devices. The sub-problem of Algorithm \ref{al:dp-homo} is to find the inference cost for a given stage, and the computational complexity is $O(nD)$. Assuming the Algorithm \ref{al:graph2stack} partitions the CNN into $L$ pieces, the total number of sub-problem of Algorithm \ref{al:dp-homo} is $O(L^2 D)$, leading to a total time complexity of $O(n L^2 D^2)$.

\textbf{Algorithm \ref{al:greedy-hete}} is a simple greedy algorithm. The sorting operation in Line 3 has $O(D \log(D))$ complexity (e.g., quick sort). 
The for-loop in Line 4 repeats for $D$ times and the line 6 has $O(\log(S))$ complexity for choosing $\mathcal{S}_{i \to j}$ from a tree set $\mathbb{S}$, where $S$ is the size of $\mathbb{S}$.
The complexity for Algorithm \ref{al:greedy-hete} is $O(D (\log(D) + \log(S)))$ and could be relaxed to $O(D \log(D))$ since $S \leq D$.

From the above discussion, we can deduce the complexity of PICO is $O(wd(\frac{nd}{w})^w + n L^2 D^2)$.
The complexity is listed in Table \ref{tab:complexity} for summary. Note Algorithm 1 only needs to run one time for every CNN and is not affected if the mobile environment changes.

\begin{table}[th]
	\centering
	\caption{Computational complexity of PICO framework.}
	\label{tab:complexity}
	\begin{adjustbox}{max width=0.95\textwidth}
		\begin{tabular}{@{}cccc@{}}
			\toprule
			Algorithm 1             & Algorithm 2    & Algorithm 3    & PICO                                \\ \midrule
			$O(wd(\frac{nd}{w})^w)$ & $O(n L^2 D^2)$ & $O(D \log(D))$ & $O(wd(\frac{nd}{w})^w + n L^2 D^2)$ \\ \bottomrule
		\end{tabular}
	\end{adjustbox}
	
\end{table}

\subsubsection{PICO in Practice}
The computational cost of PICO in practice can be decomposed into two parts: the one-time cost $O(wd(\frac{nd}{w})^w)$ caused by Algorithm \ref{al:graph2stack} and the ongoing cost $O(n L^2 D^2)$ caused by Algorithm \ref{al:dp-homo} and \ref{al:greedy-hete}.

The \textbf{one-time cost} caused by Algorithm 1 does not involve the specific edge environment or mobile device cluster, it can be executed on a powerful PC and the results can be directly used by mobile devices without additional processing. 
The \textbf{ongoing cost} caused by Algorithm 2 and 3 is lightweight and takes less than 1 second in the resource-limited Raspberry-Pi.
Thus, the Algorithm 2 and 3 can be triggered if the mobile environment changes and immediately adapt to the new environment.

\subsection{Implementation}

\begin{figure}[tb]
	\centering
	\includegraphics[width=0.8\linewidth]{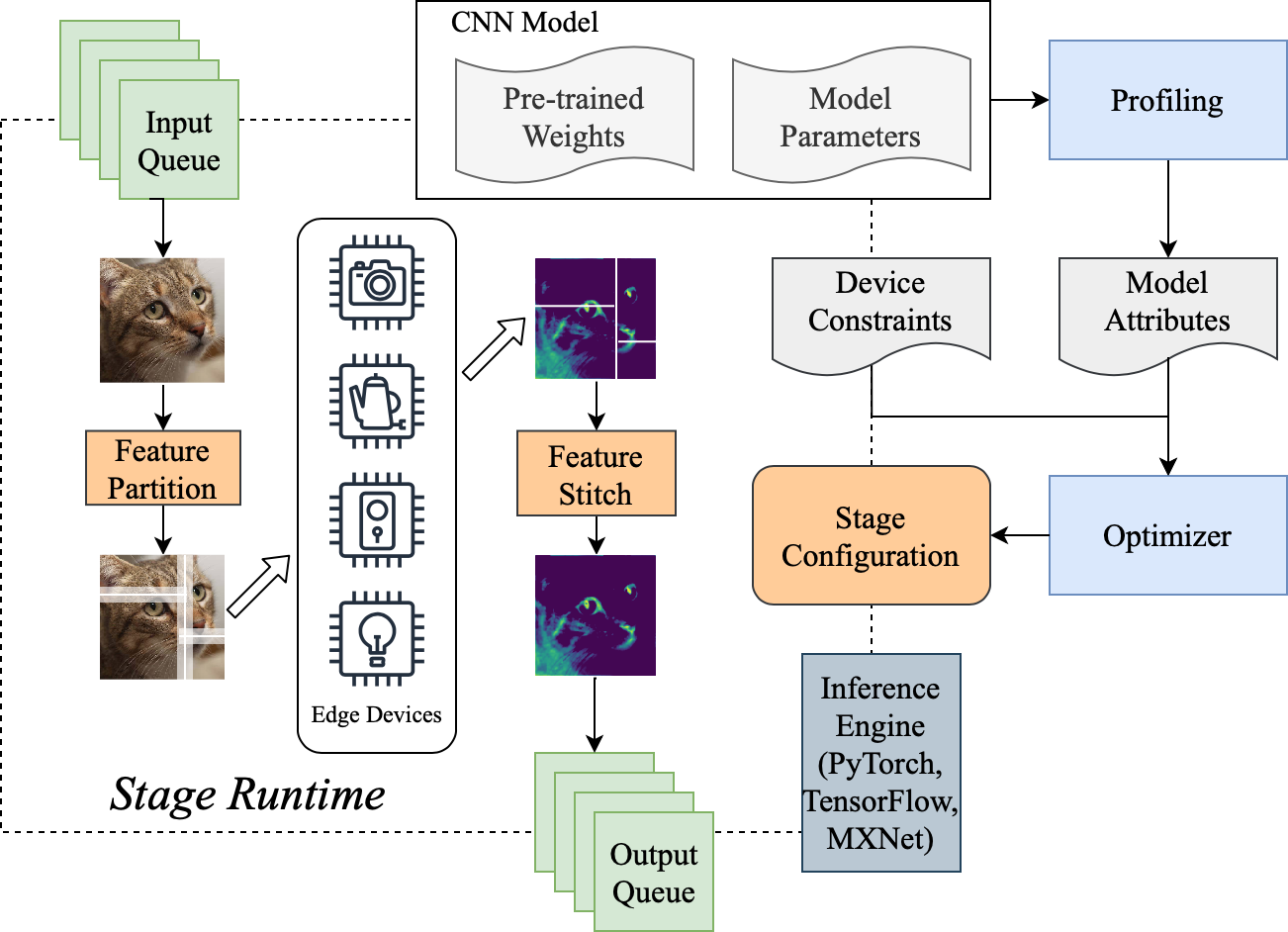}
	\caption{The workflow of stages in an inference pipeline.}
	\label{fg:sys-overview}
\end{figure}

\textbf{The workflow of stages:}
We summarize the workflow of stages in Fig. \ref{fg:sys-overview}. Each stage owns its configuration $\mathcal{S}_{i \to j}$ which is given by the previous optimization. The main thread of stage takes the feature map from the input queue, then splits it into small tiles with different sizes according to $\mathcal{F}$ and distributes them to those devices $\mathcal{D}_{i \to j}$.
Once the computation finishes, the outputs of those devices are gathered and stitched together. There are two other threads responsible to put the receiving feature map into the input queue and send the output to the next stage.

\textbf{Feature split and stitch:}
Most popular DL frameworks such as TensorFlow, PyTorch does not provide an efficient way to split feature map with overlapped parts, and using these high level provided by those frameworks to implement these operations brings intolerable latency. We accomplish the frame split and stitch operations by directly operate the frame tensor data point in the memory space through C++. In practice, after optimization, the time consumption of feature split and stitch can be ignored.

\textbf{Represent CNN into graph:}
We implement an automatic \textit{GraphConvertor} module to convert a given CNN model file into a DAG. The module will record the input and output layers for every tensor during profiling. To achieve this, We modify the source file of PyTorch and add a new hook function \textit{register\_prev\_forward\_hook} as suggested in \cite{pipedream}.

\section{Experiment} \label{sec:experiment}

We give the details of our evaluation bed for experiment and analyze the obtained result.

\subsection{Environment Setup}
Here we give the details of our evaluation setup.

\begin{figure}[tb]
	\centering
	\includegraphics[width=0.85\linewidth]{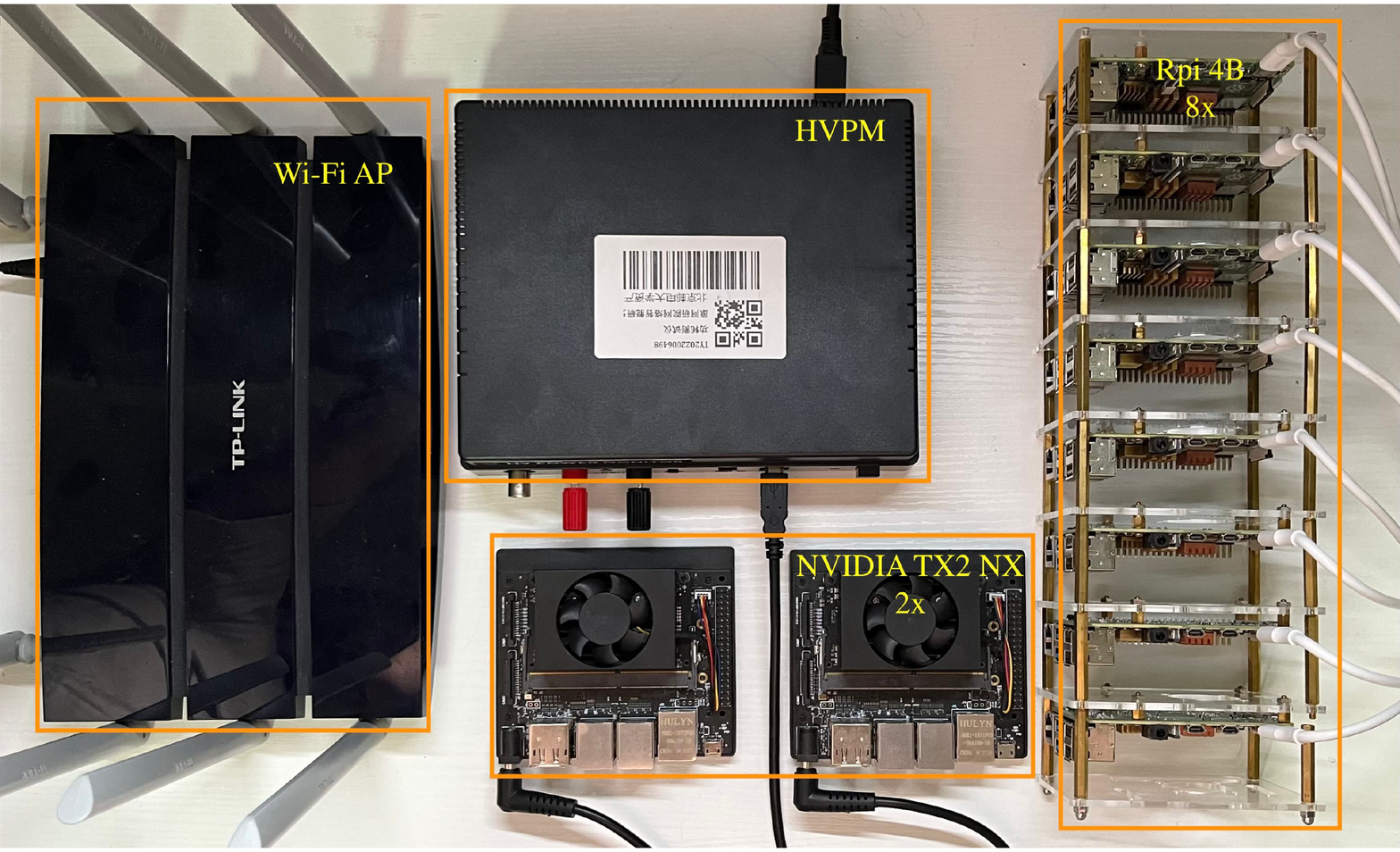}
	\caption{The testbed in our experiment is composed of 8 Raspberry-Pi 4Bs, 2 Nvidia TX2 NXs, a Monsoon High Voltage Power Monitor (HVPM) and a Wi-Fi access point.}
	\label{fg:testbed}
\end{figure}

\textbf{Hardware:}
The mobile cluster for evaluating the PICO framework uses one Wi-Fi access point with 50Mbps bandwidth and 8 ARM based Raspberry-Pi 4Bs. 
Each Raspberry-Pi 4B has a Quad Core ARM Cortex-A73, which has 1.5 GHz max CPU frequency. It has 2 GB LPDDR2 SDRAM, and dual-band 2.4 GHz/5 GHz wireless for communication. To represent a realistic low-end mobile device cluster, we set these Raspberry-Pi 4B running with one CPU core during inference. Fig. \ref{fg:testbed} shows the test bed we used, the laptop is used to monitor and control this cluster.
We limit the CPU frequency for each Raspberry-Pi using \textit{Linux cGroup} to simulate the heterogeneous mobile cluster in the real world.
The heterogeneous cluster has two Nvidia TX2 NX devices, which have a 2.2 GHz CPU frequency, and the six Raspberry-Pis that have three different CPU frequency settings: 1.5 GHz, 1.2 GHz, and 0.8 GHz, respectively.

\textbf{Software:}
We implement PICO and other compared method using PyTorch with Gloo \cite{gloo} as the communication backend. Due to differing output feature sizes on each device, we only used asynchronous P2P algorithms such as \textit{isend()} and \textit{irecv()} in Gloo. As for the network bandwidth, we use a method similar to MPEG-DASH \cite{stockhammer2011dynamic}, using the tool \textit{ping} to send data of two different sizes and measure the response times. The rate is then determined by calculating the ratio between the difference in data size and the difference in response times.

\textbf{Models overview:}
VGG16 \cite{simonyan2014very} is a classic CNN classification model. It contains 13 conv layers, 5 pooling layers and 3 fc layers. 
You only look once version 2 (YOLOv2) \cite{redmon2017yolo9000} is a lightweight CNN used for real-time object detection system. It has deeper architecture compared with VGG-16. There are 23 conv and 5 pooling layers in YOLOv2, nearly twice as VGG16.
Both VGG16 and YOLOv2 follow the chain structure.
ResNet34 \cite{resnet} and InceptionV3 \cite{szegedy2017inception} are two classic CNNs that use a block structure. ResNet34 uses a \textit{skip-connection} strategy that allow the feature to skip several layers. Compared with ResNet34, InceptionV3 has more complex structure. The Inception block has multiple branches, and the conv layers also have many unbalanced (e.g., $1\times 7$, $5 \times 1$) kernels. 

\textbf{Compared method:}
For VGG16 and YOLOv2, four different parallelization strategies are used in the evaluation: (1) Layer-wise (LW) scheme, which parallelizes the CNNs layer by layer; (2) Early-fused-layer (EFL) scheme, an extension to the implementation of DeepThings \cite{zhao2018deepthings}, which fuses and parallelizes the first few conv layers, then executes the rest layers in a single device; (3) Optimal Fused-layer (OFL) scheme, which selectively fuses convolution layers at different parts of a model; (4) CoEdge (CE) \cite{coedge}, the state-of-the-art parallelization scheme which extends the layer-wise scheme to heterogeneous environment and reduce the communication overhead by sorting devices into a list and limiting the communication to the neighbors for each device. Moreover, CoEdge uses a dynamic number of devices to process different layers to further reduces the impact of communication overhead instead of using all mobile devices. (5) Pipeline Cooperation (PICO) scheme, which is proposed in this paper. For ResNet34 and YOLOv2, we compare two different graph partition strategies: (1) Consider each block as a piece, which is used in \cite{zhou2019adaptive,deepslicing}. (2) Partition the entire graph into multiple pieces with suitable granularity, which is proposed in this paper.

\begin{figure}[tb]
	\centering
	\includegraphics[width=0.85\linewidth]{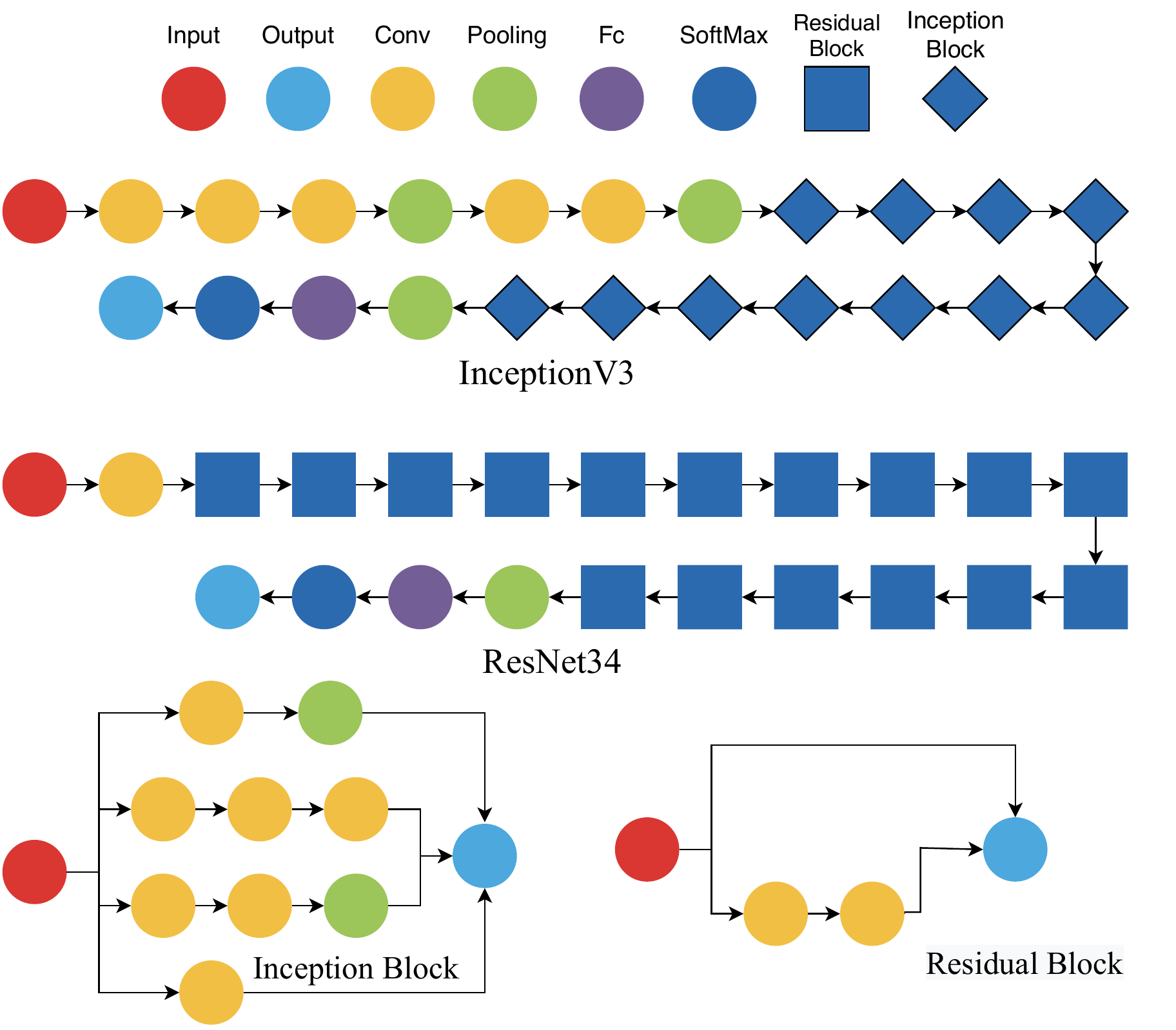}
	
	\caption{The model structure of ResNet34 and InceptionV3.}
	\label{fg:graph-dnn}
\end{figure}

\begin{figure}
	\centering
	\subfloat[Original]{ \label{fg:partition-pieces-before}
		\includegraphics[width=0.43\linewidth]{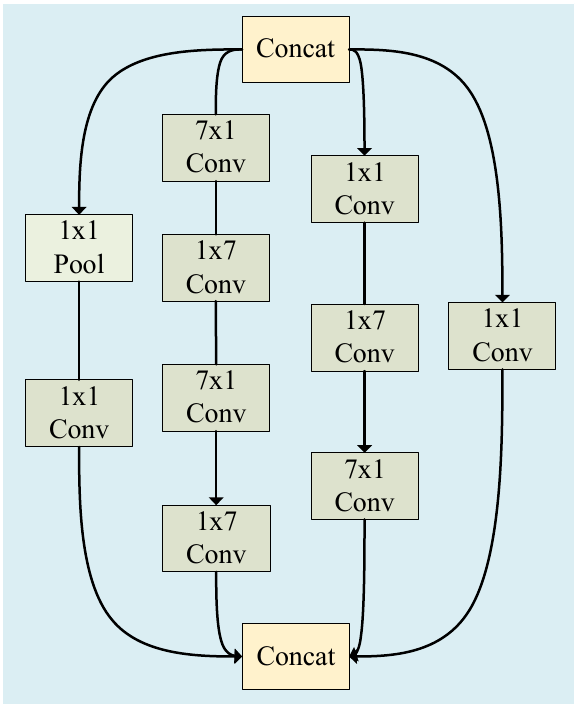}
	}
	\subfloat[Optimized]{ \label{fg:partition-pieces-after}
		\includegraphics[width=0.54\linewidth]{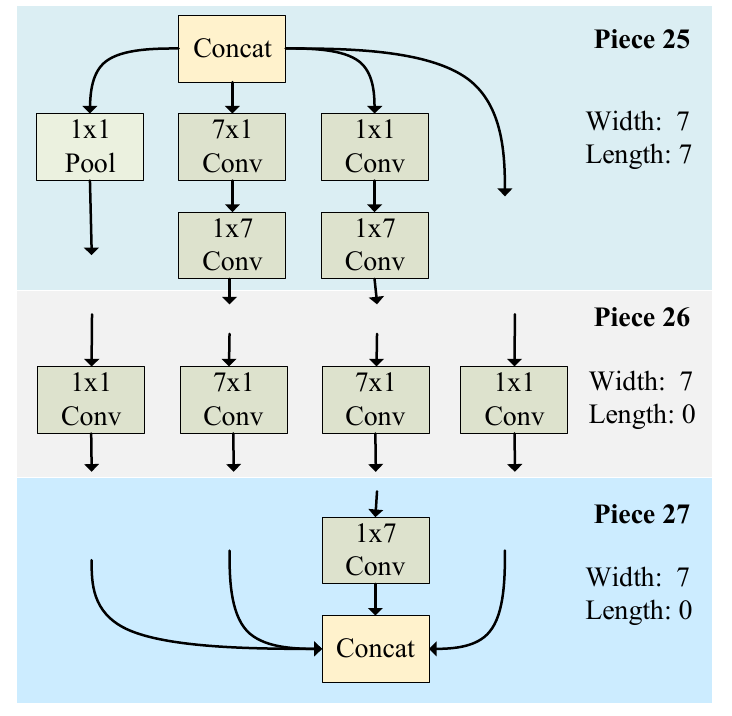}
	}
	\caption{
		An illustration of our graph partition algorithm.
		(a) Part of InceptionV3 model (InceptionC block).
		(b) The obtained pieces after optimization.}
	\label{fg:partition-pieces}
\end{figure}

\subsection{CNN Graph Partition}
Here we present some experimental results of our proposed graph partition algorithm.

\begin{figure}[tb]
	\centering
	\subfloat[ResNet34]{
		\includegraphics[width=0.9\linewidth]{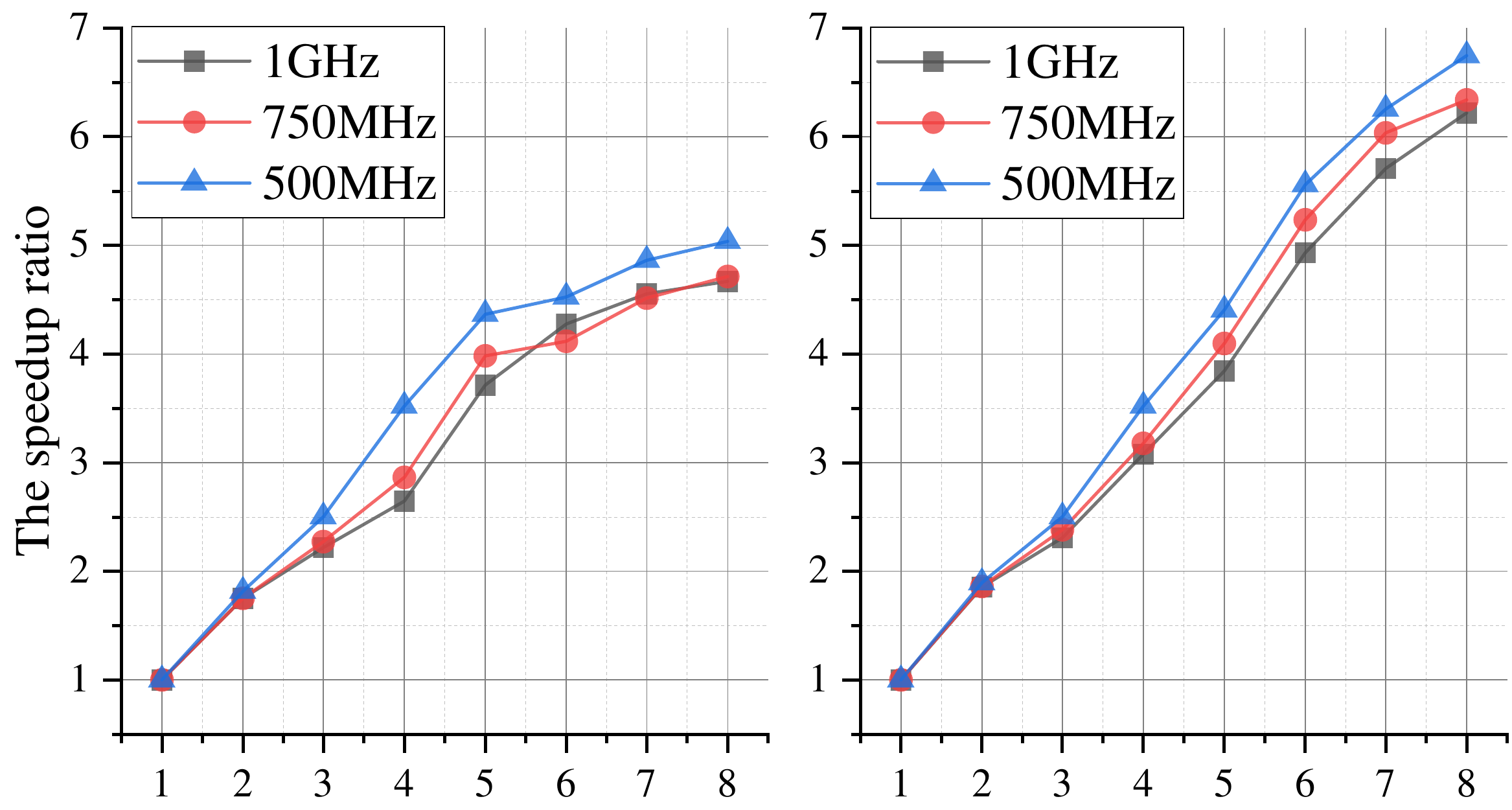}
	}
	
	\subfloat[InceptionV3]{
		\includegraphics[width=0.9\linewidth]{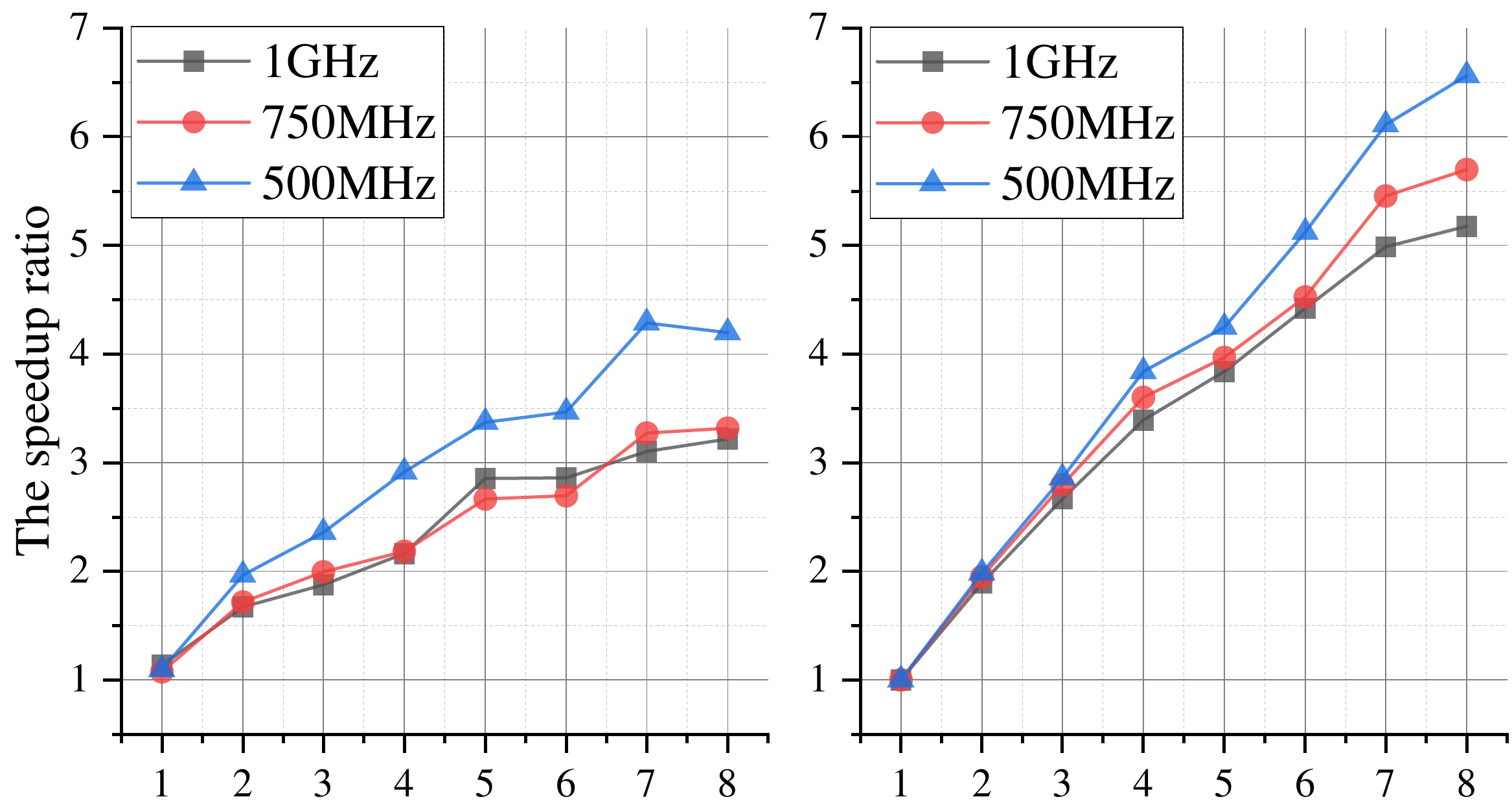}
	}
	\caption{The speedup ratio for ResNet34 and InceptionV3. The left part shows the result by treating the entire block as a whole, and the right part uses graph partition algorithm.}
	\label{fg:graph-dnn-speedup}
\end{figure}

\subsubsection{Partitioned Pieces}
Fig. \ref{fg:partition-pieces} shows part of the partition result of InceptionV3 model. Fig. \ref{fg:partition-pieces-before} plots the InceptionC block, which consists $4$ branches and $10$ neural layers.
We can find if we consider the entire block as a layer \cite{zhou2019adaptive}, lots of redundant calculations will be introduced since there are many \textit{unbalanced} conv kernels (e.g., $1 \times 7$ and $7 \times 1$).
We can use Eq. \eqref{eq:necessary-wihi} and Eq. \eqref{eq:to-next-wh} to quantify the redundant calculation.
If we fuse the entire block into one piece (used in \cite{zhou2019adaptive}), the devices have to introduce $13$ pixel length on both the width and height dimensions.
After running the partition algorithm, the block is split into three pieces (Piece $25$, Piece $26$, and Piece $27$) as shown in Fig. \ref{fg:partition-pieces-after} (The full partition result is attached in the supplemental material).
The entire InceptionV3 model is composed of $40$ pieces, but due to the size of model, we can not plot the entire model here. The complete partition result is shown in the supplemental materials.
These pieces have much smaller redundant calculation. Piece $25$ has $7$ pixel length redundancy, and Piece $26$ and Piece $27$ only have redundancy on only one dimension. Compared with Fig. \ref{fg:partition-pieces-before}, the redundant calculation during inference can be significantly reduced.
In addition, since we break block into pieces, the later optimization can make more fine-grained optimization.

\subsubsection{Speedup After Partition}
We can adapt PICO to those CNNs with non-chain structure by applying our graph partition algorithm at first. Here we compare the speedup ratio for ResNet34 and InceptionV3. Fig. \ref{fg:graph-dnn} shows the structures of the two model, obviously they are constructed with the block structure.
According to the figure, we can find the Inception block in InceptionV3 is more complex than Residual block used in ResNet34.
Fig. \ref{fg:graph-dnn-speedup} plots the speedup ratio under different CPU frequencies for ResNet34 and InceptionV3 with different strategy. 
The figures on the left part fuse the entire block into a whole, and the figures on the right show the results that adopt our graph partition algorithm,
When executing CNN inference with 8 devices, PICO can achieve $6.8 \times$ speedup for ResNet34 and $6.5 \times$ for InceptionV3 after partitioning the CNN model. The speedup effect is more obvious with low CPU frequency since the limitation of computing resource is relieved when the number of mobile devices increases. 
As for the strategy of fusing the entire block, it achieves up to $5 \times$ speedup for ResNet34, but only $4 \times$ for InceptionV3 when the CPU frequency is low.
We think it is caused by the difference in the number of layers in Residual and Inception blocks. Since the Inception block contains more layers than Residual block, the optimal model partition is more likely to exist within blocks.

\subsubsection{Optimization Complexity}
\begin{table}[t]
	\centering
	\begin{tabular}{@{}cccccc@{}}
		\toprule
		Model       & $n$            & $w$ & $wd(\frac{nd}{w})^w$   & Execution & Pieces \\ \midrule
		VGG16       & 19             & 1   & $4.7 \times {10}^{2}$  & 0.10s     & 19     \\
		SqueezeNet  & 30             & 2   & $5.6 \times {10}^{4}$  & 0.14s     & 29     \\
		ResNet34    & 38             & 2   & $9.0 \times {10}^{4}$  & 0.28s     & 21     \\
		MobileNetV3 & 96             & 3   & $6.1 \times {10}^{7}$  & 0.79s     & 31     \\
		InceptionV3 & 99             & 4   & $4.6 \times {10}^{9}$  & 3.01s     & 38     \\
		NASNetL     & 570            & 8   & $1.1 \times {10}^{22}$ & $>$ 5h    & NaN    \\
		NASNetL-P   & $75 \times  8$ & 8   & $9.3 \times {10}^{14}$ & 1.9h      & 34     \\
		\bottomrule
	\end{tabular}
	\caption{The performance of Algorithm \ref{al:graph2stack} for various CNNs. NASNetL-P denotes the strategy which roughly partitions it into 8 parts.}
	\label{tab:algo1-cnn}
\end{table}

\begin{figure*}[tb]
	\centering
	\subfloat[Inference period under different parallel schemes and CPU frequencies]{ \label{fg:vgg-throughput-period}
		\includegraphics[width=0.72\linewidth]{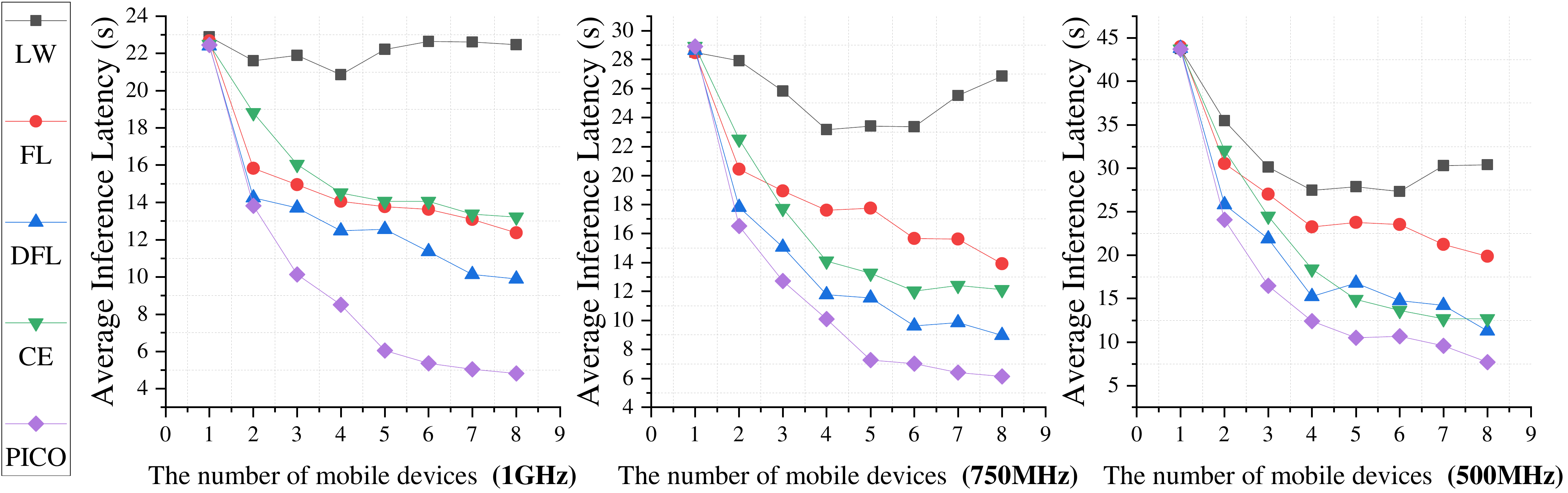}
	}
	\subfloat[Throughput with 8 devices]{
		\includegraphics[width=0.24\linewidth]{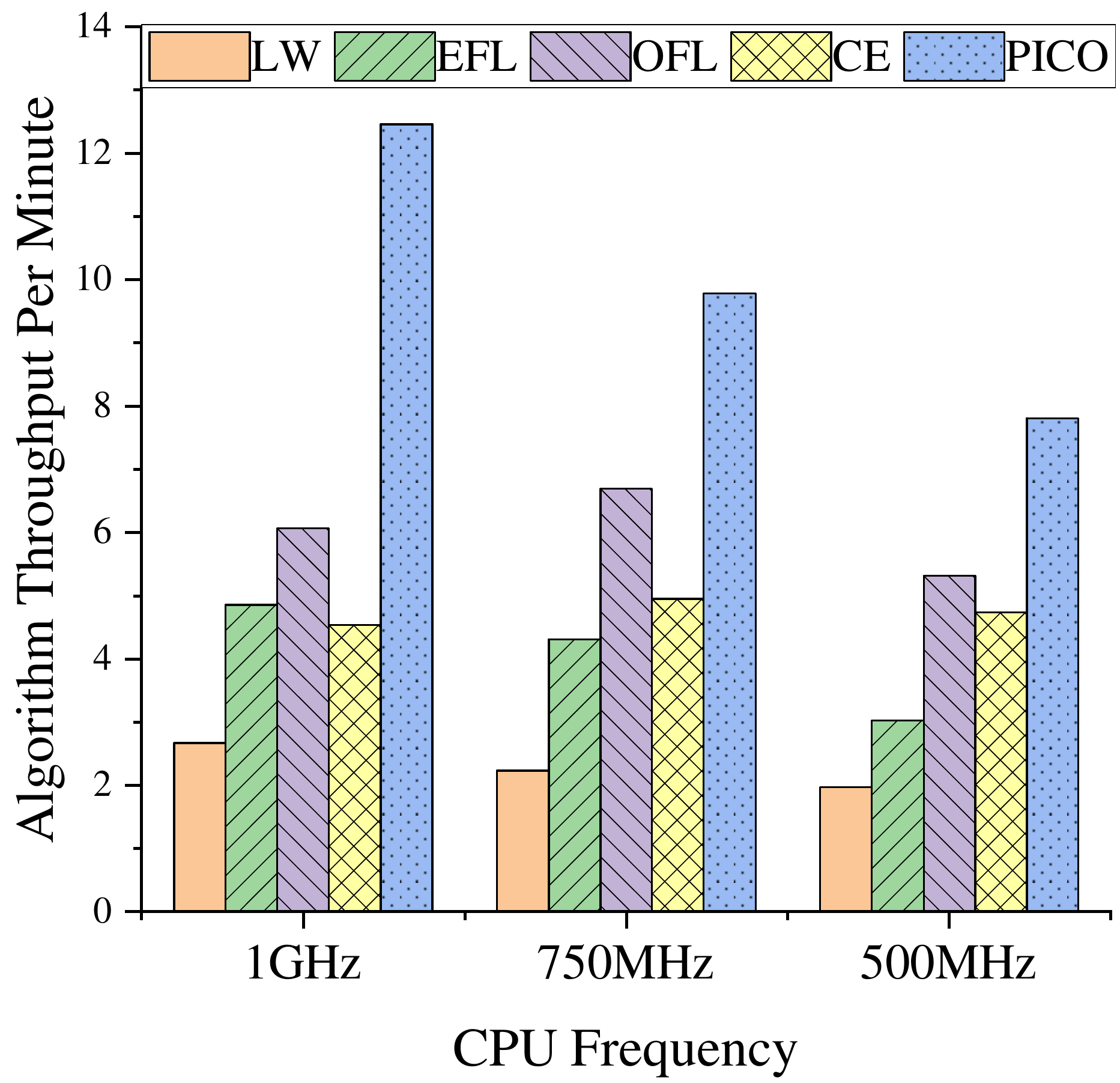}
	}
	\caption{The cluster capacity when executing VGG16.}
	\label{fg:vgg-throughput}
\end{figure*}

\begin{figure*}[tb]
	\centering
	\subfloat[Inference period under different parallel schemes and CPU frequencies]{
		\includegraphics[width=0.72\linewidth]{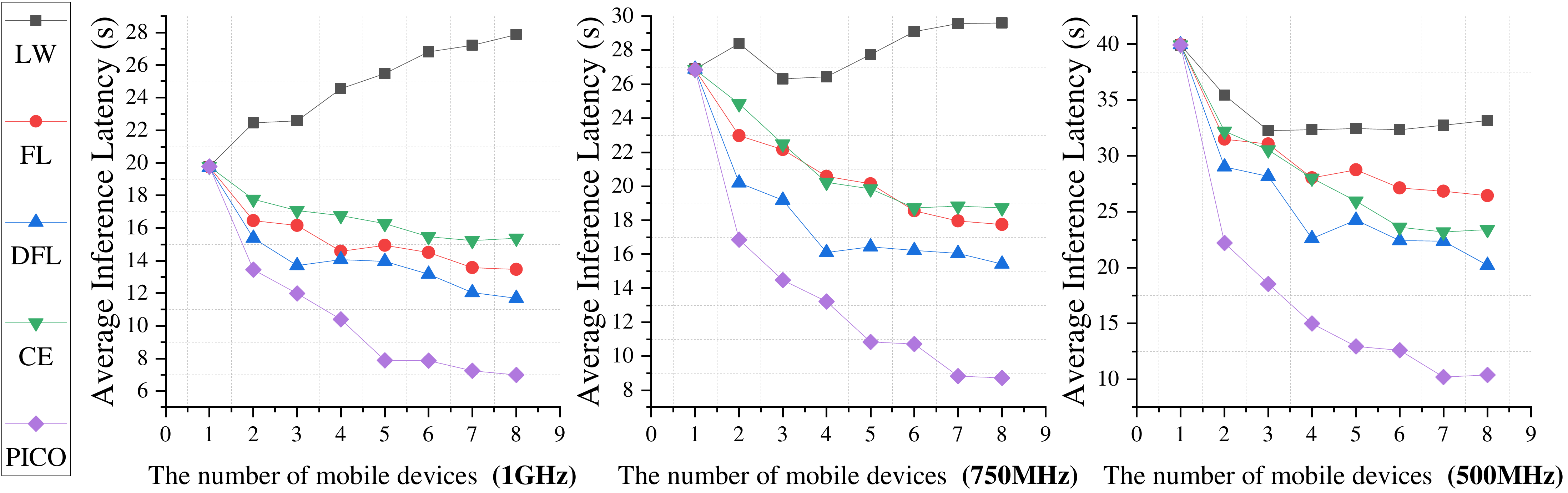}
	}
	\subfloat[Throughput with 8 devices]{\includegraphics[width=0.24\linewidth]{
			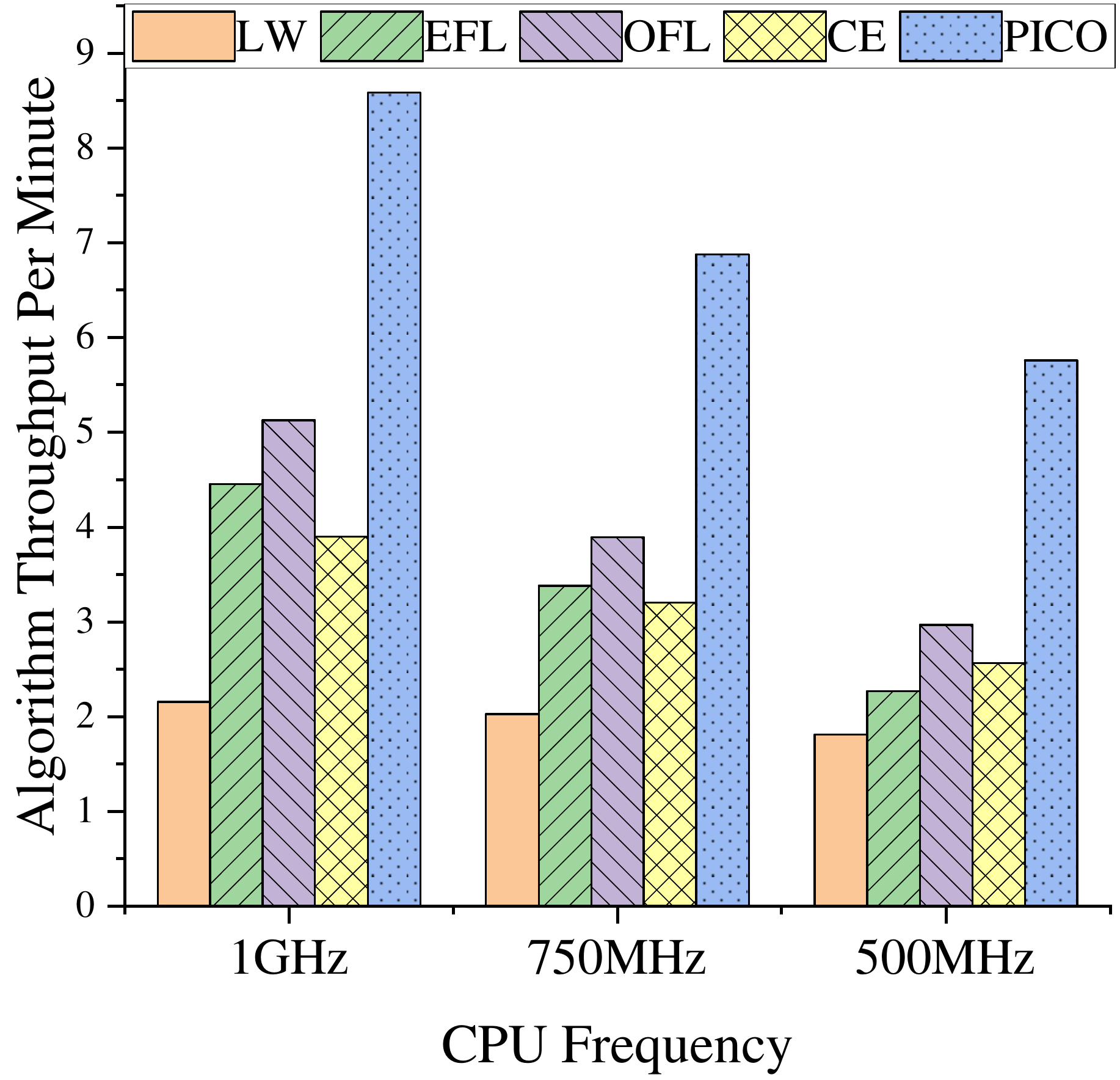}
	}
	\caption{The cluster capacity when executing YOLOv2.}
	\label{fg:yolo-throughput}
\end{figure*}

\begin{figure}
	\centering
	\includegraphics[width=0.95\linewidth]{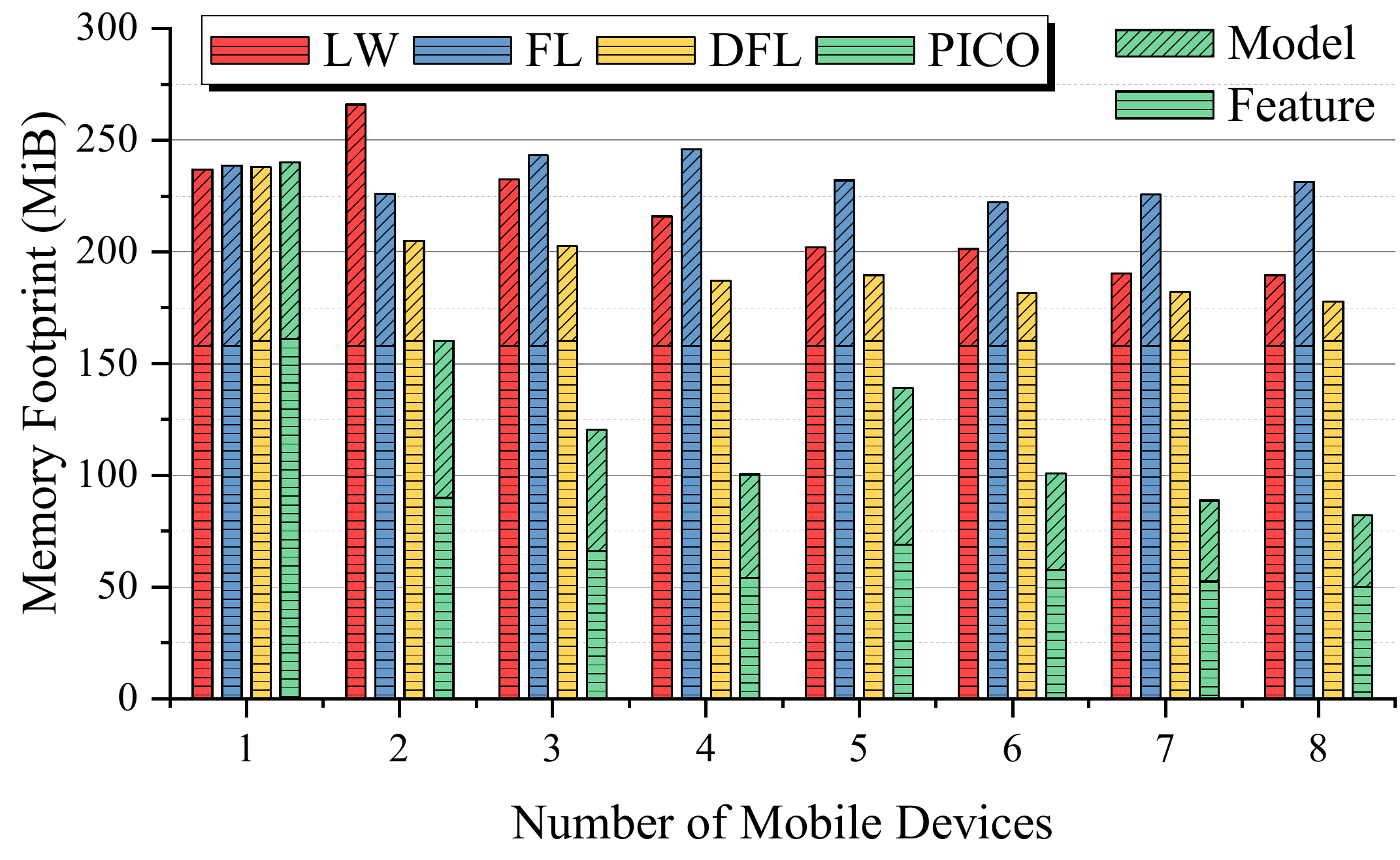}
	
	\caption{The memory footprints of different algorithms.}
	\label{fg:graph-memory}
\end{figure}

Algorithm \ref{al:graph2stack} has $O(wd(\frac{nd}{w})^w)$ complexity to optimize the given CNN, as we analyzed in Section \ref{sec:complexity-analysis}.
Here we run Algorithm \ref{al:graph2stack} on many popular CNNs on a PC equipped with Intel Core i9-10940X to give a comprehensive evaluation of its performance.
The number of layers $n$, the width $w$ and the upper bound $wd(\frac{nd}{w})^w$ for every tested CNN and the execution time are listed in Table \ref{tab:algo1-cnn}. 
Note $n$ only counts \textit{conv} and \textit{pool} layers, since other layers such as \textit{BN} and \textit{ReLU} do not change the output feature shape and require negligible computing resources.
Additionally, the last column shows the number of pieces after optimization.

Many real-world CNNs are deep but narrow, which means they have small width $w$.
Algorithm \ref{al:graph2stack} is efficient and could be executed in less than one or several seconds for these models.
We also add the NASNet-A-Large \cite{zoph2018learning} (NASNetL) model to evaluate Algorithm \ref{al:graph2stack} in an extremely complex case.
NASNetL is constructed through \textit{neural architecture search} technology. NASNetL is much broader ($w=8$) and contains much more layers ($n=570$) compared with the hand-designed models ($w \leq 4$ and $n \leq 100$).
It is rare to deploy such a large-scale model on mobile devices. The trade-off that considers a block as a layer \cite{zhou2019adaptive,deepslicing} has no effect since there is no block in NasNetL.

Directly applying PICO to NASNetL takes nearly infinite time to produce the final output considering the time complexity ($1.1 \times 10^{22}$).
We successfully adapt PICO to NASNetL using \textit{divide-and-conquer}.
Assume Algorithm \ref{al:graph2stack} divides a model into $L$ pieces, if we fuse the $L/2$ pieces from the input position into a smaller model and apply Algorithm \ref{al:graph2stack} to it, then the smaller model must be divided into the same $L/2$ pieces as the original model (the property of dynamic programming).
Inspired by this property, We cut a small part from the beginning of NASNetL, and apply Algorithm \ref{al:graph2stack} on the smaller model to obtain several pieces. Only these pieces away from the cut line will be kept to make sure these pieces from different small model are still sequential.
Then we apply the same strategy to the rest model until all the smaller models are partitioned into pieces.
The last line in Table \ref{tab:algo1-cnn} shows the performance of the divide-and-conquer strategy. NasNetL is tackled with 8 parts and PICO produces the result in two hours.
Since Algorithm \ref{al:graph2stack} only needs to run once for every CNN regardless of specific mobile environment (see Section \ref{sec:complexity-analysis}), the optimization cost is acceptable.

\begin{table*}[tb]
	\caption{The utilization, redundancy ratios and memory footnotes of heterogeneous mobile devices. }
	\label{tab:heter}
	\begin{adjustbox}{max width=\textwidth}
		\begin{tabular}{@{}llllrrrrrrrrr@{}}
			\toprule
			\multirow{2}{*}{Model}   & \multirow{2}{*}{Attributes}          & \multirow{2}{*}{Methods}       & \multirow{2}{*}{Type} & \multicolumn{8}{c}{Devices} & \multirow{2}{*}{Average}                                                                       \\ \cmidrule(lr){5-12}
			                         &                                      &                                &                       & NX@2.2                      & NX@2.2                   & Rpi@1.5 & Rpi@1.5 & Rpi@1.2 & Rpi@1.2 & Rpi@0.8 & Rpi@0.8 &         \\ \cmidrule(r){1-13}
			\multirow{12}{*}{VGG16}  & \multirow{12}{*}{\shortstack{Layers:                                                                                                                                                                                         \\13 conv \\ + \\ 5 pool \\ \\Input size: \\ 244 $\times$ 244}}
			                         & \multirow{3}{*}{CE}                  & Utili.                         & 80.87\%               & 82.13\%                     & 69.37\%                  & 59.97\% & 57.91\% & 36.56\% & 23.11\% & 17.33\% & 53.40\%           \\
			                         &                                      &                                & Redu.                 & 2.02\%                      & 1.93\%                   & 1.29\%  & 2.06\%  & 1.30\%  & 1.41\%  & 1.32\%  & 0.77\%  & 1.51\%  \\
			                         &                                      &                                & Mem.                  & 195.0 M                     & 183.0 M                  & 162.0 M & 158.0 M & 147.0 M & 149.0 M & 134.0 M & 137.0 M & 158.1 M \\
			\cmidrule{3-13}
			                         &                                      & \multirow{3}{*}{EFL}           & Utili.                & 32.43\%                     & 39.79\%                  & 72.58\% & 75.08\% & 94.23\% & 96.77\% & 64.09\% & 64.16\% & 67.39\% \\
			                         &                                      &                                & Redu.                 & 11.02\%                     & 11.60\%                  & 19.08\% & 19.83\% & 18.58\% & 19.22\% & 12.78\% & 13.42\% & 15.69\% \\
			                         &                                      &                                & Mem.                  & 142.0 M                     & 147.0 M                  & 169.0 M & 179.0 M & 173.0 M & 183.0 M & 151.0 M & 165.0 M & 163.6 M \\
			\cmidrule{3-13}
			                         &                                      & \multirow{3}{*}{OFL}           & Utili.                & 38.90\%                     & 40.19\%                  & 60.87\% & 61.79\% & 85.34\% & 94.15\% & 76.54\% & 80.46\% & 67.28\% \\
			                         &                                      &                                & Redu.                 & 7.45\%                      & 7.67\%                   & 11.12\% & 11.39\% & 10.33\% & 10.53\% & 8.15\%  & 8.31\%  & 9.36\%  \\
			                         &                                      &                                & Mem.                  & 149.0 M                     & 149.0 M                  & 158.0 M & 159.0 M & 154.0 M & 155.0 M & 151.0 M & 152.0 M & 153.4 M \\
			\cmidrule{3-13}
			                         &                                      & \multirow{3}{*}{\textbf{PICO}} & Utili.                & 91.21\%                     & 93.28\%                  & 83.12\% & 79.40\% & 47.63\% & 66.26\% & 68.17\% & 90.15\% & 77.40\% \\
			                         &                                      &                                & Redu.                 & 11.08\%                     & 10.97\%                  & 5.82\%  & 3.83\%  & 6.93\%  & 5.55\%  & 0.00\%  & 3.83\%  & 6.00 \% \\
			                         &                                      &                                & Mem.                  & 189.0 M                     & 144.0 M                  & 121.0 M & 103.0 M & 92.0 M  & 121.0 M & 115.0 M & 111.0 M & 124.5 M \\
			\cmidrule(l){1-13}
			\multirow{12}{*}{YOLOv2} & \multirow{12}{*}{\shortstack{Layers:                                                                                                                                                                                         \\23 conv \\ + \\ 5 pool \\ \\Input size: \\ 448 $\times$ 448}}
			                         & \multirow{3}{*}{CE}                  & Utili.                         & 76.85\%               & 75.46\%                     & 65.81\%                  & 66.94\% & 46.32\% & 46.77\% & 22.49\% & 20.21\% & 52.61\%           \\
			                         &                                      &                                & Redu.                 & 0.82\%                      & 0.76\%                   & 0.87\%  & 0.83\%  & 0.79\%  & 0.71\%  & 0.68\%  & 0.61\%  & 0.75\%  \\
			                         &                                      &                                & Mem.                  & 265.0 M                     & 260.0 M                  & 255.0 M & 246.0 M & 245.0 M & 240.0 M & 235.0 M & 239.0 M & 248.1 M \\
			\cmidrule{3-13}
			                         &                                      & \multirow{3}{*}{EFL}           & Utili.                & 37.85\%                     & 35.64\%                  & 67.24\% & 67.61\% & 96.01\% & 95.28\% & 75.87\% & 72.81\% & 68.54\% \\
			                         &                                      &                                & Redu.                 & 27.09\%                     & 27.09\%                  & 45.08\% & 45.08\% & 44.68\% & 44.68\% & 29.29\% & 29.29\% & 36.54\% \\
			                         &                                      &                                & Mem.                  & 189.0 M                     & 178.0 M                  & 208.0 M & 208.0 M & 208.0 M & 207.0 M & 178.0 M & 178.0 M & 194.3 M \\
			\cmidrule{3-13}
			                         &                                      & \multirow{3}{*}{EFL}           & Utili.                & 39.28\%                     & 37.03\%                  & 69.47\% & 68.92\% & 97.02\% & 95.99\% & 77.61\% & 73.94\% & 69.91\% \\
			                         &                                      &                                & Redu.                 & 25.98\%                     & 25.98\%                  & 44.51\% & 44.51\% & 44.86\% & 44.86\% & 28.12\% & 28.12\% & 35.86\% \\
			                         &                                      &                                & Mem.                  & 193.0 M                     & 182.0 M                  & 212.0 M & 212.0 M & 212.0 M & 211.0 M & 182.0 M & 182.0 M & 198.3 M \\
			\cmidrule{3-13}
			                         &                                      & \multirow{3}{*}{\textbf{PICO}} & Utili.                & 89.37\%                     & 97.91\%                  & 89.96\% & 97.85\% & 89.44\% & 99.40\% & 91.89\% & 89.03\% & 93.11\% \\
			                         &                                      &                                & Redu.                 & 6.95\%                      & 2.27\%                   & 1.25\%  & 9.18\%  & 9.18\%  & 5.89\%  & 6.13\%  & 5.05\%  & 5.73\%  \\
			                         &                                      &                                & Mem.                  & 188.0 M                     & 135.0 M                  & 108.0 M & 116.0 M & 113.0 M & 122.0 M & 159.0 M & 157.0 M & 137.3 M \\
			\bottomrule
		\end{tabular}
	\end{adjustbox}
\end{table*}

\subsection{Pipeline Performance}

We evaluate our proposed pipeline cooperation scheme with 2-8 Raspberry-Pi devices and measure some important metrics. 

\subsubsection{Maximum Throughput}
Fig. \ref{fg:vgg-throughput} and Fig. \ref{fg:yolo-throughput} plot the cluster capacity when executing VGG16 and YOLOv2 with different parallel schemes. The first three figures plot the inference period with different parallel schemes and CPU frequencies. The last figure plots the accomplished inference task per minute with 8 devices. It represents the throughput of different parallel schemes.
PICO has the best performance as expected, since our optimization goal is to reduce the redundant computation and achieve minimum pipeline period. 
When the number of devices increases, the throughput of different strategies also improve except the executing YOLOv2 using LW with 1GHz CPU core. 
YOLOv2 has nearly twice number of layers compared with VGG16, which brings more communication overhead for Layer-wise strategy. 
Through CE also executes CNN layer by layer, CE uses a dynamic number of working devices during inference and reduces the traffic volume by only synchronizing overlapped features. Therefore, CE outperforms LW.
When the computing resource is rich (1GHz), the gain brought by the increasing number of devices is offset by communication overhead. 
EFL and OFL fuse multiple layers into one model segment, and do not require communication among devices when they are executing one segment, thus the communication overhead is reduced.
Since OFL optimizes the configuration of fused layers, it outperforms EFL which simply fuses the very early layers.
However, when the number of devices is bigger than a certain number (4 for example), the improvement is very tiny due to the additional computation CPU redundancy.

\subsubsection{Memory Footprint}
Memory footprint is another important metric during inference. The inference latency will quickly grow when the required memory exceeds the onboard memory of the device, since the device has to use swap memory \cite{deepslicing}.
We use a python script to sample memory footprint from \textit{/proc/pid/status} for each inference process. Fig. \ref{fg:graph-memory} plots the average memory footprint of different algorithms. Here we ignore the performance of CE, since LW and CE have very similar performance when devices are homogeneous. According to our previous discussion, the memory footprint can be divided into two more fine-grained parts. The \textit{Model} and \textit{Feature} denote how much the model parameters and intermediate features take part in the memory footprint. 
We can find the memory footprint decreases as the number of mobile devices increases. Since LW, FL, DFL only split features, the whole model needs to be replicated on all mobile devices. 
This approach leads to the result that they can only decrease the memory footprint caused by intermediate features. Meanwhile, PICO distributes both models and features, thus PICO reduces the memory footprint significantly.

\begin{figure}
	\centering
	\includegraphics[width=0.9\linewidth]{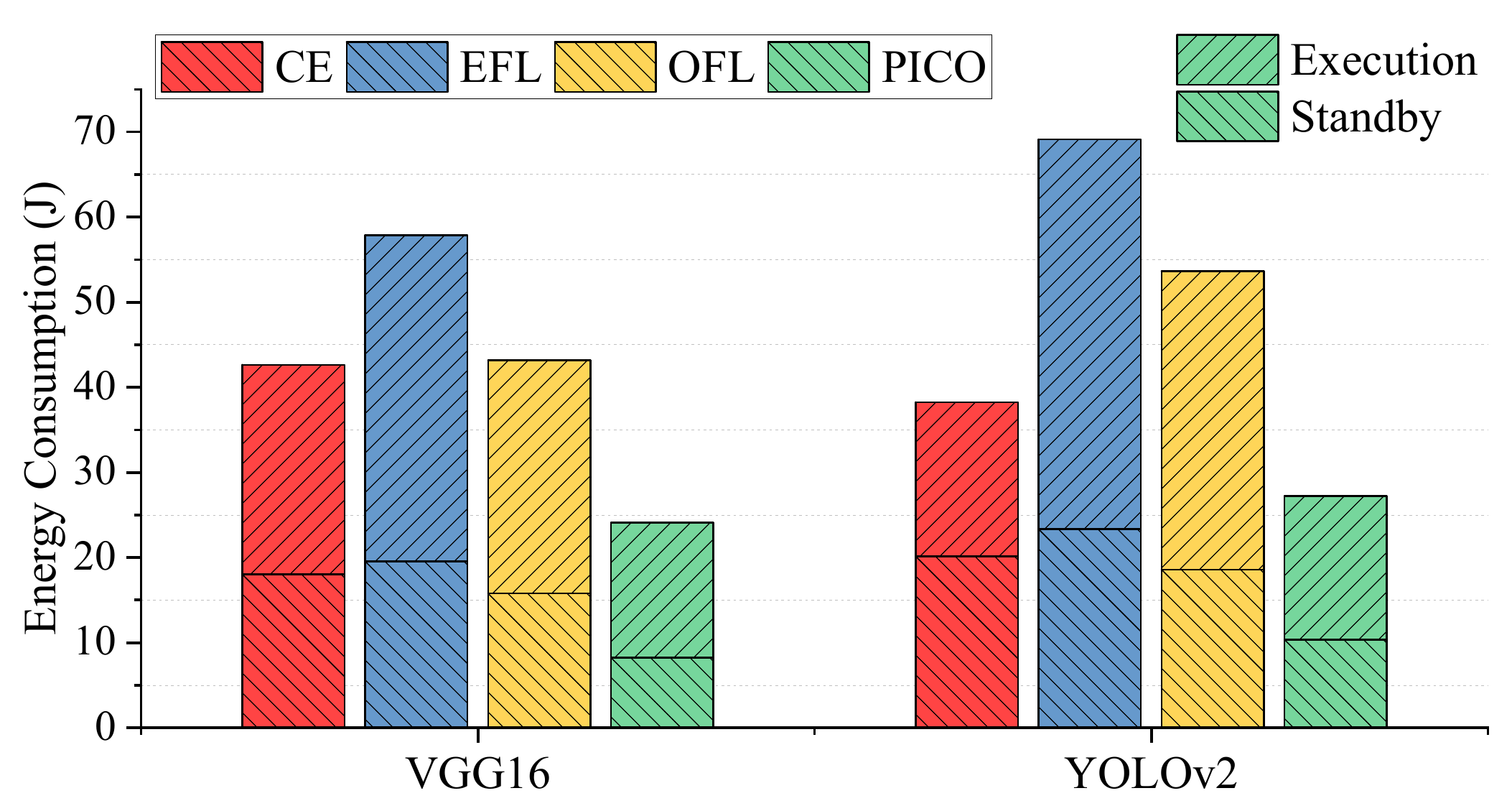}
	\caption{The average energy consumption for every inference task with heterogeneous devices.}
	\label{fg:power}
\end{figure}

\subsection{Impact of Heterogeneity}
Here we evaluate the impact of heterogeneity on different parallel schemes.
We monitor the CPU usage during inference on the heterogeneous mobile cluster and record the average computing resource utilization ratio (Utili.) for different parallel schemes. 
We also calculate the redundancy ratio (Redu.) and their memory footprint (Mem.) on every device during computation. The result is presented in Table \ref{tab:heter}. We remove LW scheme due to its bad performance in heterogeneous environment.

\subsubsection{Load Balancing}
The CPU utilization rates of these devices show the workload among these devices. PICO and CE will adjust the feature partition size according to the specific devices. Thus, the workload of PICO and CE is better than EFL and OFL.
We find both PICO and CE impose more percentages of workloads on these devices with higher CPU frequency (1.2 GHz). Take the CE as an example, the CPU utilization rate is up to 82.61\% for the fastest devices, but drifts down 22.64\% for the slowest devices when running VGG16.
The reason is that CE uses a dynamic number of devices to process each layer. When the feature map is wide (e.g., 224 x 224), CE may use all devices to accelerate the execution. When the feature map is small (e.g., 7 x 7), CE may place all the workload on one powerful device to avoid redundant computation and communication.
However, the computing resources of these slower devices are wasted. On the contrary, PICO can fully utilize the computing resources, thus having a better performance on load balancing.

\subsubsection{Computation Efficiency}
Because the input feature maps of different devices overlap with each other, the redundant computation can lead to inefficient performance.
CE has the minimum average redundant computation, since CE synchronizes the feature map for every layer. But the frequent communication leads to low resource utilization and high inference latency. 
Fusing layers and executing them together can keep the devices busy, but will increase redundant computation. Especially for the EFL which has 46.54\% percent redundancies executing YOLOv2. OFL uses dynamic programming to find a balance between communication and computation, but the redundancy ratio (12.08\%) is still higher than PICO (7.64\%) as PICO uses a subset of mobile devices instead of the entire cluster.

\subsubsection{Energy Consumption}
We measure the energy consumption for every inference task, the result is shown in Fig. \ref{fg:power}. The energy consumption is composed of the inference execution and standby power consumption.
EFL consumes the most energy, since EFL has the highest redundant computation compared with other schemes. Moreover, the redundant computation does not accelerate the inference, thus EFL also has high standby power consumption.
OFL has a lower energy consumption compared with EFL since OFL reduces the redundant computation by synchronizing feature map periodically.
CE executes the CNN layer by layer and has the lowest redundancy among all the schemes. However, the standby power consumption is the majority of energy consumption, because CE has a long inference latency, especially executing YOLOv2.
On the contrary, PICO has the lowest standby power consumption during inference task, since PICO can maximize the throughput during inference.
Through PICO has more redundant computation compared with CE, the overall energy consumption is still lower than CE.

\subsection{Comparing With Optimal Configuration}

Because it is NP-Hard to find the best many-to-many mapping for graph-like CNN and heterogeneous devices, PICO can not guarantee finding the optimal inference pipeline configuration.
Thus, we compare PICO with the optimal pipeline to further evaluate the performance.
The optimal pipeline is obtained through a broad first search (BFS). We compare the optimization time for producing the pipeline configuration and the resource utilization of every mobile device during runtime.

\subsubsection{Methodology}

The main problem for the comparison is the possible solution space for BFS is over-complex. 
According to Table \ref{tab:optim-complex}, finding the best many-to-many mapping for both chain-like CNN, heterogeneous devices and graph-like CNN, homogeneous devices are NP-Hard.
But BFS tries to find the best many-to-many mapping for graph-like CNN and heterogeneous devices.
We test the BFS with CNNs on 4-8 Raspberry-Pi devices, but all of them fail to produce the final output after several hours on a powerful PC.
Therefore, we compare the performance of PICO and BFS from two sides.
On the one side, (1) we compare PICO and BFS with graph-like CNN and homogeneous devices. On the other side, (2) we compare PICO and BFS with chain-like CNN and heterogeneous devices.
Table \ref{tab:graph-homo} and \ref{tab:chain-heter} show the optimization overhead of PICO and BFS.
Fig. \ref{fg:graph-homo} and \ref{fg:chain-heter} give the runtime performance.

\begin{table}[tb]
	\caption{Optimization time with graph-like CNN. }
	\centering
	\begin{tabular}{c r r}
		\toprule
		\textbf{Branches, Layers, Devices} & \textbf{PICO} & \textbf{BFS (Optimal)} \\
		\midrule
		(2, 4, 4)                          & $< 1$s        & $1.58$s                \\
		(2, 8, 6)                          & $< 1$s        & $18.23$s               \\
		(3, 12, 4)                         & $< 1$s        & $11.96$m               \\
		(3, 12, 6)                         & $< 1$s        & $45.24$m               \\
		(3, 12, 8)                         & $< 1$s        & $> 1$s                 \\
		(4, 20, 4)                         & $< 1$s        & $> 1$h                 \\
		(4, 20, 6)                         & $< 1$s        & $> 1$h                 \\
		\bottomrule
	\end{tabular}
	\label{tab:graph-homo}
\end{table}

\begin{table}[tb]
	\caption{Optimization time with heterogeneous devices.}
	\centering
	\begin{tabular}{c r r}
		\toprule
		\textbf{Layers, Devices} & \textbf{PICO} & \textbf{BFS (Optimal)} \\
		\midrule
		(4, 4)                   & $< 1$s        & $< 1$s                 \\
		(8, 4)                   & $< 1$s        & $1.62$s                \\
		(12, 4)                  & $< 1$s        & $3.84$s                \\
		(16, 4)                  & $< 1$s        & $11.27$s               \\
		(8, 6)                   & $< 1$s        & $4.35$m                \\
		(10, 6)                  & $< 1$s        & $12.28$m               \\
		(12, 6)                  & $< 1$s        & $> 1$h                 \\ 
		(8, 8)                   & $< 1$s        & $> 1$h                 \\
		\bottomrule
	\end{tabular}
	\label{tab:chain-heter}
\end{table}

\begin{figure}[t]
	\centering
	\includegraphics[width=0.95\linewidth]{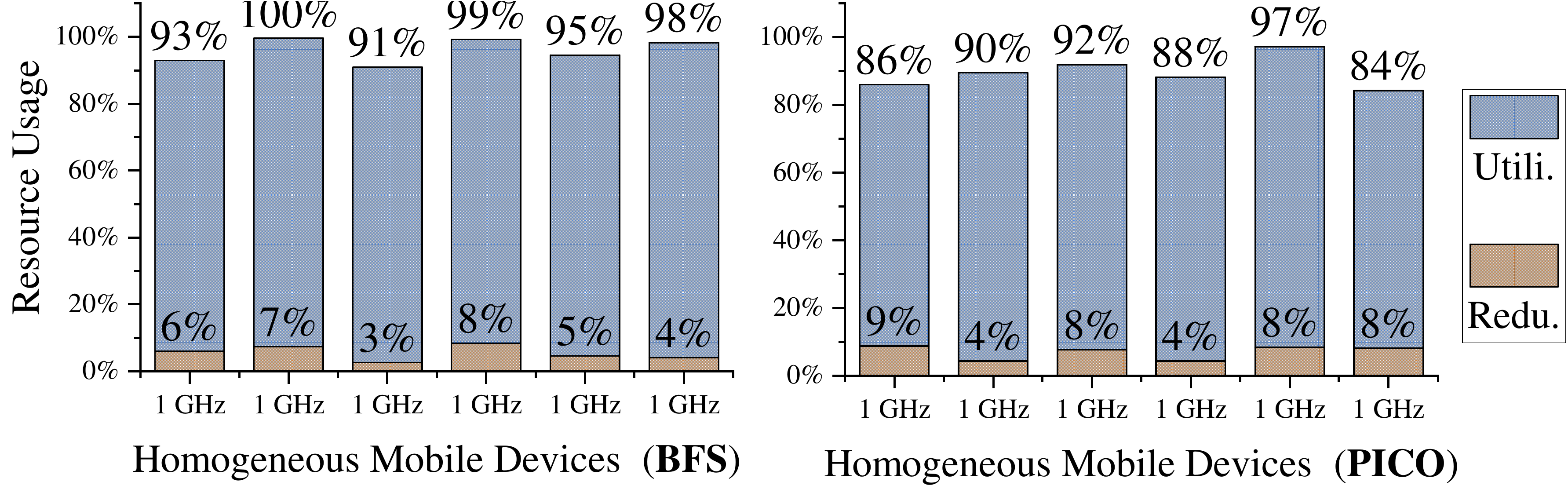}
	\caption{Runtime performance with graph-like CNN.}
	\label{fg:graph-homo}
\end{figure}
\begin{figure}[t]
	\centering
	\includegraphics[width=0.95\linewidth]{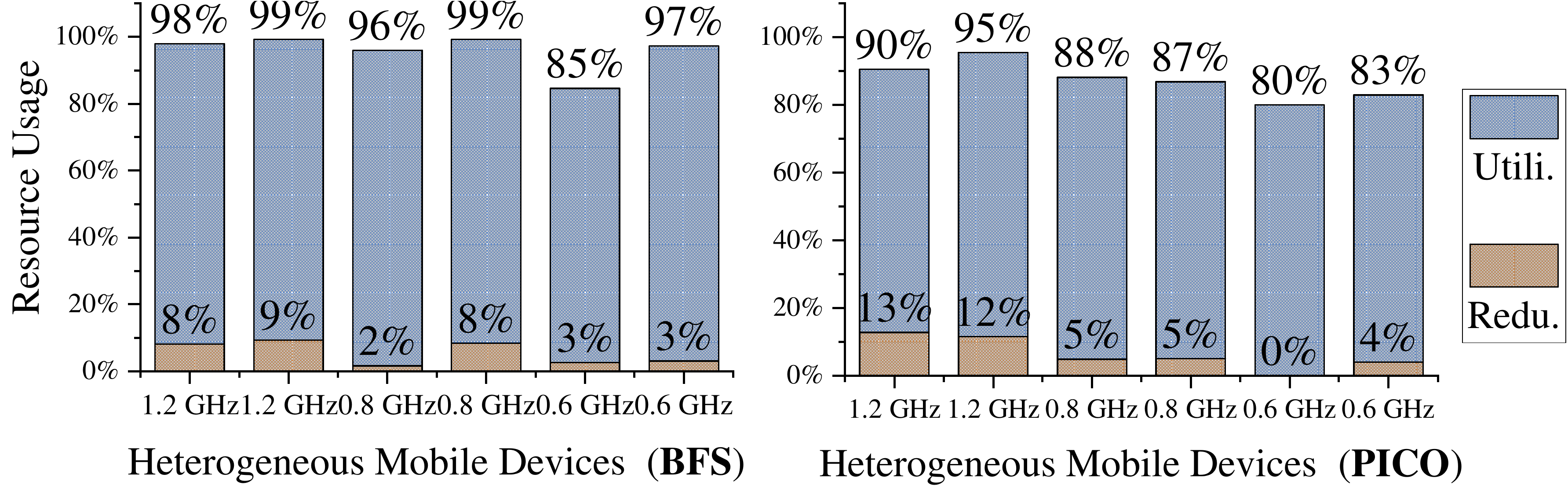}
	\caption{Runtime performance with heterogeneous devices.}
	\label{fg:chain-heter}
\end{figure}

\subsubsection{Optimization Time}
For all the situations listed in Table \ref{tab:graph-homo} and \ref{tab:chain-heter}, PICO could accomplish the optimization within 1 second, But BFS requires much more time to give the output even on small scale problems. The optimization time dramatically grows on larger problems and BFS fails to finish the calculation on both sides.
Moreover, these problems that BFS fails to solve are much easier (either chain-like CNN or homogeneous devices) than those that PICO has solved in the paper. Thus, BFS is not applicable in practice.

Another observation from Table \ref{tab:graph-homo} and \ref{tab:chain-heter} is that the changing of different parameters (branches, layers, devices) has different impacts on the optimization time. When the CNN is a graph and devices are homogeneous, increasing the number of layers has more impact than devices.
Take the Table \ref{tab:graph-homo} as an example, row 3 and row 4 show that the optimization time increases from $11.96$ minutes to $45.24$ minutes when the number of devices increases from $4$ to $6$ ($3.78 \times$).
But row 2 and row 3 show that the optimization time increases from $18.23$ seconds to $11.96$ minutes when the number of layers increases from $8$ to $12$ ($39.36 \times$) through the number of homogeneous devices decreases.
On the contrary, increasing the number of devices has more impact when the CNN is a chain and devices are heterogeneous, as shown in Table \ref{tab:chain-heter}.
The optimization time increases from $1.62$ seconds to $3.84$ seconds ($2.37 \times$) when the number of layers increases from $4$ to $8$ (row 1 and row 2), but it increases from $1,62$ seconds to $4.35$ minutes ($161.11 \times$) when the number of devices increases from $4$ to $6$ (row 2 and row 5).
These two observations reveal the complexity of the many-to-many mapping when the CNN is a graph and devices are heterogeneous from sides.

\subsubsection{Runtime Performance}
We compare the runtime performance of PICO and BFS by plotting the computing resources utilization rate for each device. The result is plotted in Fig. \ref{fg:graph-homo} and \ref{fg:chain-heter}.
We also analyze the redundant computation during inference since high utilization rate does not lead to good performance \cite{zhou2019adaptive}.

Fig. \ref{fg:graph-homo} shows the runtime performance for a graph-like CNN and 6 homogeneous devices (1 GHz CPU frequency). The graph-like CNN used in the comparison contains 3 branches and 12 layers and is also used in row 4, Table \ref{tab:graph-homo}.
The optimal configuration found by BFS achieves 95\% resource utilization rate. Meanwhile, the configuration found by PICO has the similar performance (around 90\%).
The redundant computation of BFS is lower than PICO, but all redundant computation keep at a low level for both BFS and PICO.
The performance for chain-like CNN and heterogeneous devices is plotted in Fig. \ref{fg:chain-heter}. The CNN contains 10 layers and these devices have different computing resources, as shown on the x-axis (1.2 GHz, 0.8 GHz and 0.6 GHz).
Similar to Fig \ref{fg:graph-homo}, the optimal configuration (BFS) achieves great performance on these devices (up to 99\%) except one (85\%). 
As for PICO, the configuration places more workload of the inference to these devices who own rich computing resources, thus the resource utilization of them is similar to BFS (90\% and 95\% for the fastest devices). The average performance of the other devices is around 84.5\%.
Since PICO greatly reduces the computation complexity according to previous analysis, the performance of PICO is acceptable for most real world applications.

\section{Related Work} \label{sec:related}
Along with the problem of enabling DNN-based intelligent applications, previous researches can be divided into two categories.

\subsection{Inference Offloading}
Due to the limited up-link of mobile devices, traditional way of uploading captured data to the cloud server is time-consuming \cite{ko2018edge,li2018jalad}. Researchers focus on offloading the computation of early layers to mobile devices (\textit{Inference offloading}).
To minimize the inference latency, Neurosurgeon \cite{kang2017neurosurgeon} proposed to partition model between cloud server and mobile device according to the network situation. But \cite{kang2017neurosurgeon} can only handle models with the chain structure.
DADS \cite{hu2019dynamic} proposed a novel algorithm to partition DNN with graph structure using a min-cut algorithm. QDMP \cite{zhang2020towards} noticed that directly applying min-cut on the entire graph is time-consuming. Based on the block structure, \cite{zhang2020towards} proposed a divide-and-conquer algorithm to find the min-cut, which achieves a nearly linear complexity in their experiments.
Meanwhile, Branchynet \cite{teerapittayanon2016branchynet} propose \textit{early exit} mechanism by adding exit layers at the midden of DNN. This mechanism enables mobile device not feature map to cloud server if the local accuracy already reaches a certain value. Considering the situation when server does not have the corresponding model,  IoNN \cite{jeong2018ionn} an incremental offloading technology that significantly improves the inference performance.

\subsection{Cooperative Inference}
Recently researchers began to turn their attentions on executing inference completely at the edge with multiple mobile devices \cite{mao2017modnn,mao2017mednn,zhao2018deepthings,zhou2019adaptive,hadidi2020towards,coedge,edgeflow,deepslicing,distredge}.

MoDNN \cite{mao2017modnn} is the first work in this field. MoDNN equally partitions the out feature map for every layer and distributes these feature maps to homogeneous devices. In their following-up work MeDNN \cite{mao2017mednn}, they use an adaptive partition method for the heterogeneous devices. Both MoDNN and MeDNN need a master device to gather the entire output of every device for every layer. CoEdge \cite{coedge} reduces the communication overhead by only sending the overlapped feature map to the neighbors of devices. CoEdge also dynamically adjusts the number of working devices during inference to find the balance between communication and computation. EdgeFlow \cite{edgeflow} introduces a forwarding table to overlap the communication with computation for CNNs with complex structures. The devices can execute one layer and receive the feature map required by other layers at the same time. All these works \cite{mao2017modnn,mao2017mednn,coedge,edgeflow} require devices to communicate with each other for every layer. However, the wireless environment can lead to considerable communication overhead using these works.

Deepthings \cite{zhao2018deepthings} proposed to fuse the layers in the early stage of CNN to avoid communication during inference. But fusing layer increases overlapped feature maps among devices and harms the inference efficiency.
DistrEdge \cite{distredge} trains a deep reinforcement learning model to distribute the inference workload for heterogeneous devices. AOFL \cite{zhou2019adaptive} uses a dynamic programming to find a trade-off between communication and computation. Devices need to synchronize the feature map after several layers using AOFL. DeepSlicing \cite{deepslicing} propose a runtime scheduler to distribute the workload for heterogeneous devices. Both AOFL and DeepSlicing partition the CNN at the block level. Moreover, all these works \cite{zhao2018deepthings,zhou2019adaptive,deepslicing} are at a loss for what to do when meeting some extremely complex CNN \cite{cvpr18:nashnet}. On the contrary, PICO breaks the block into smaller pieces to avoid additional redundant computation.

\section{Conclusion and Further Research} \label{sec:conclusion}

In this paper, we propose a pipeline cooperation scheme (PICO) for efficiently executing inference with versatile CNN models and diverse mobile devices. This scheme improves the inference efficiency by reducing the redundant calculation. We first analyze the problem of partitioning CNNs and mobile devices into an inference pipeline. Using the analysis result, PICO uses a two-step strategy to build the pipeline. First, we orchestrate the graph structure of the given CNN into a sequence of pieces. Then we divide these pieces and devices into several stages. The input data is fed into the first stage and the inference result is produced at the last stage. These stages compose an inference pipeline. We adjust the partition size of features among devices according to their computing resources. The execution time of each stage is optimized to be as close as possible to gain maximum throughput. In our experiment with 8 Raspberry-Pi devices, the throughput can be improved by $1.8 \sim 6.8 \times$ under various settings.

PICO has demonstrated strong performance across a range of heterogeneous clusters by adjusting the partitioned feature size for each device to accommodate varying computation capabilities. However, this approach is limited in addressing device-level imbalances within a given stage and is unable to address imbalances at the stage-level. This can result in failure if the computation capabilities of the devices are extremely varied. To address these challenges, we are actively pursuing the development of a novel algorithm that can better balance the workload across different stages. This is a critical area of focus for our ongoing research efforts.

\appendices

\ifCLASSOPTIONcompsoc

	\section*{Acknowledgments}
\else
	
	\section*{Acknowledgment}
\fi

\ifCLASSOPTIONcaptionsoff
	\newpage
\fi
We gratefully acknowledge the support received for this work from several sources. This includes the National Natural Science Foundation of China (Grants 62071067, 62171057, 62201072), the Ministry of Education and China Mobile Joint Fund (MCM20200202), Beijing University of Posts and Telecommunications-China Mobile Research Institute Joint Innovation Center, and BUPT Excellent Ph.D. Students Foundation (CX2021134). Additionally, we appreciate the funding from the Key-Area Research and Development Program of Guangdong Province (No. 2021B0101400003), Hong Kong RGC Research Impact Fund (No. R5060-19), Areas of Excellence Scheme (AoE/E-601/22-R), General Research Fund (No. 152203/20E, 152244/21E, 152169/22E), and Shenzhen Science and Technology Innovation Commission (JCYJ20200109142008673).

\bibliographystyle{IEEEtran}
\bibliography{reference}

\vspace{-5em}
\begin{IEEEbiography}[
		{\includegraphics[width=1in,height=1.25in,clip,keepaspectratio]{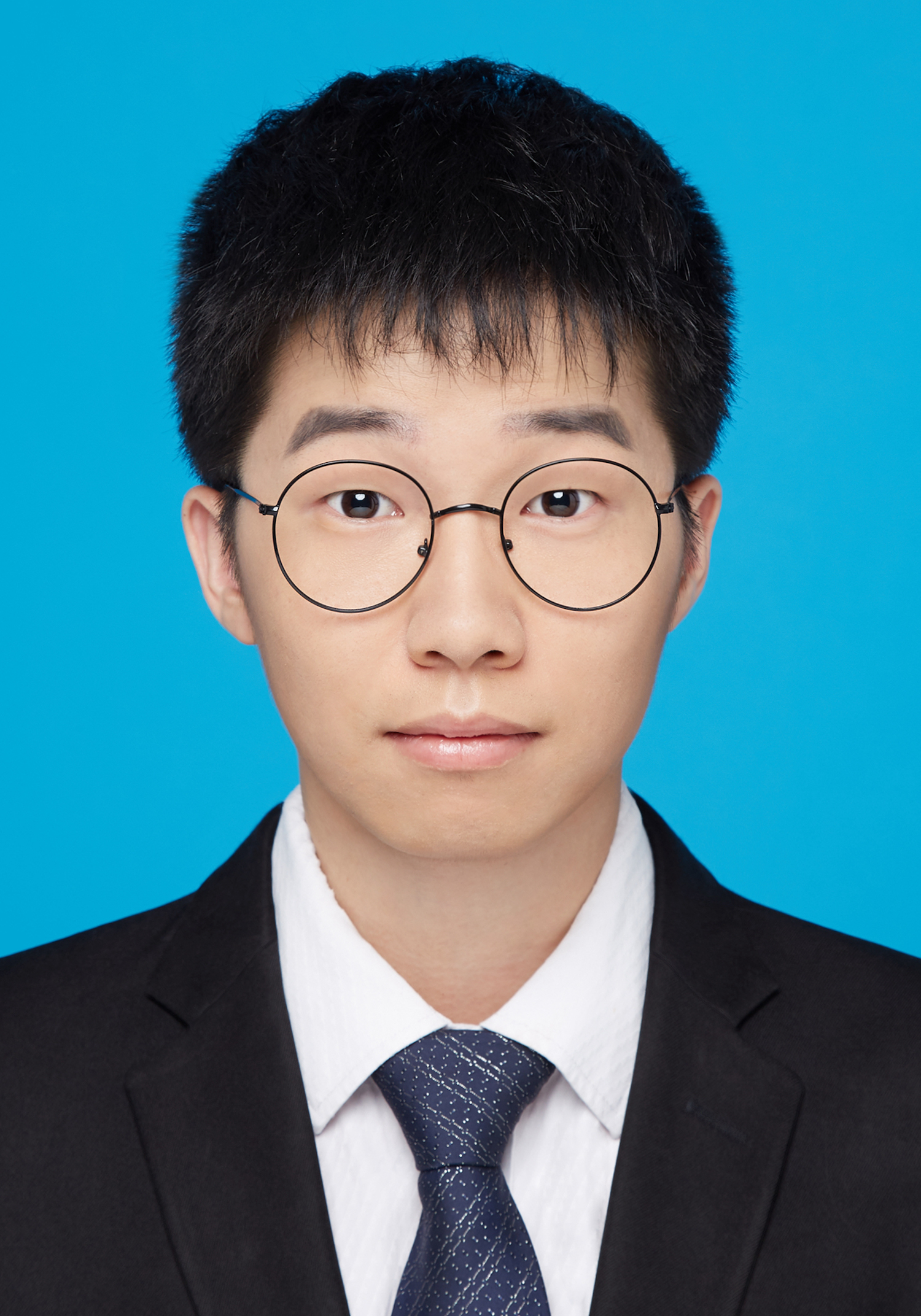}}]{Xiang Yang} received the B.E. degree in computer science and technology from Beijing University of Posts and Telecommunications, Beijing, China, in 2019. He is currently a PhD candidate of State Key Laboratory of Networking and Switching Technology at Beijing University of Posts and Telecommunications. His research interests span broad aspects of machine learning, distributed computing, edge/cloud computing and deep learning.
\end{IEEEbiography}
\vspace{-5em}

\begin{IEEEbiography}[
		{\includegraphics[width=1in,height=1.25in,clip,keepaspectratio]{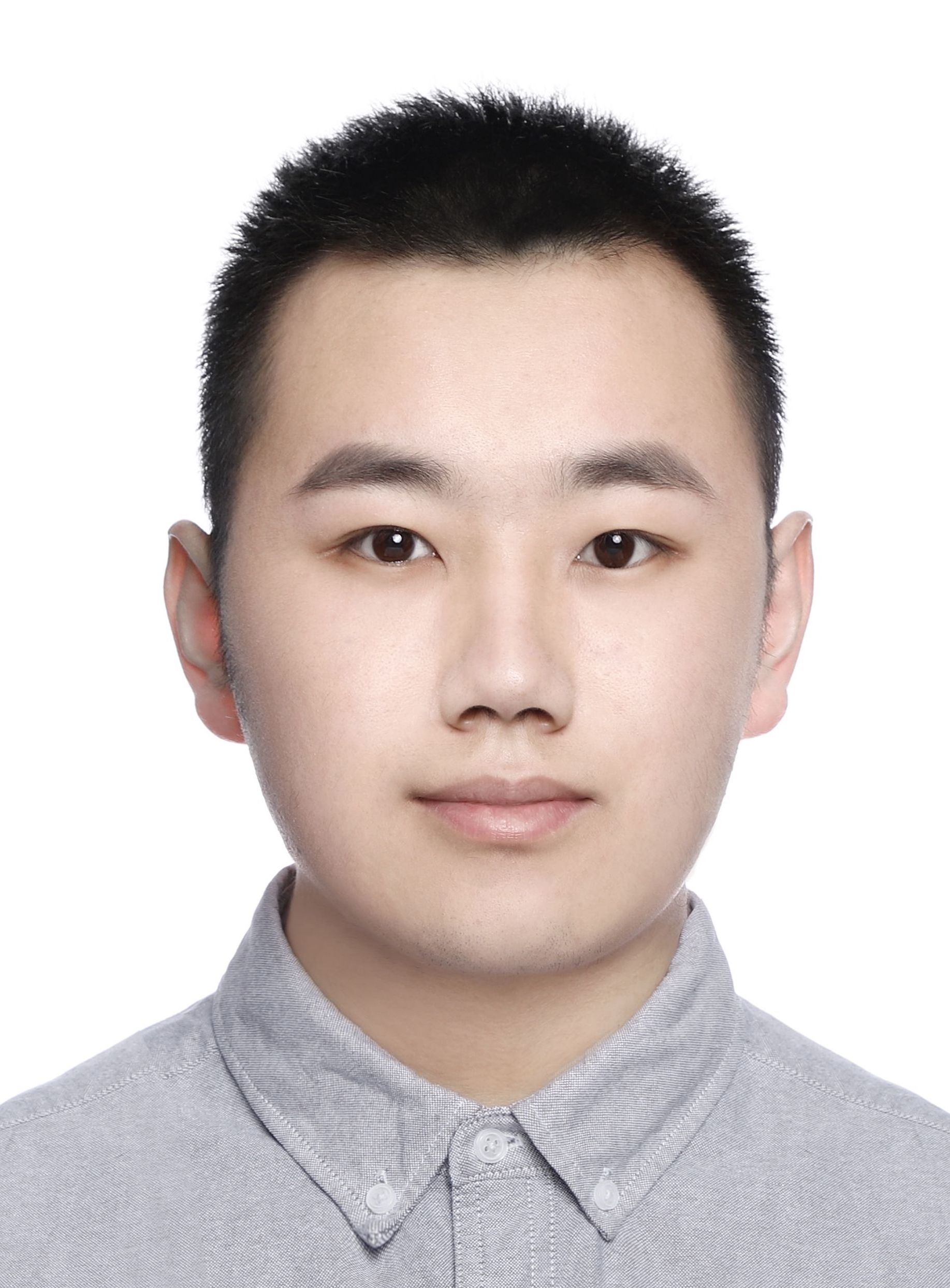}}]{Zikang Xu} is an undergraduate student majoring in Computer Science and Technology at Beijing University of Posts and Telecommunications. He obtained a postgraduate recommendation of State Key Laboratory of Networking and Switching Technology at Beijing University of Posts and Telecommunications. His research interests span broad aspects of machine learning, edge/cloud computing and distributed computing.
\end{IEEEbiography}

\vspace{-2em}

\begin{IEEEbiography}[
		{\includegraphics[width=1in,height=1.25in,clip,keepaspectratio]{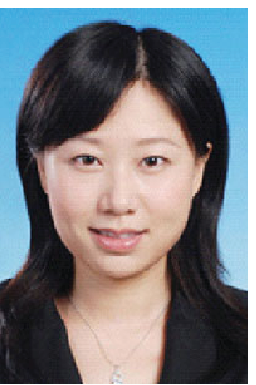}}]{Qi Qi} obtained her PhD degree from Beijing University of Posts and Telecommunications in 2010. Now, she is an associate professor of State Key Laboratory of Networking and Switching Technology at Beijing University of Posts and Telecommunications. She has published more than 30 papers in international journal, and obtained two National Natural Science Foundations of China. Her research interests include edge computing, mobile cloud computing, Internet of Things, ubiquitous services, deep learning, and deep reinforcement learning.
\end{IEEEbiography}
\vspace{-2em}

\begin{IEEEbiography}[
		{\includegraphics[width=1in,height=1.25in,clip,keepaspectratio]{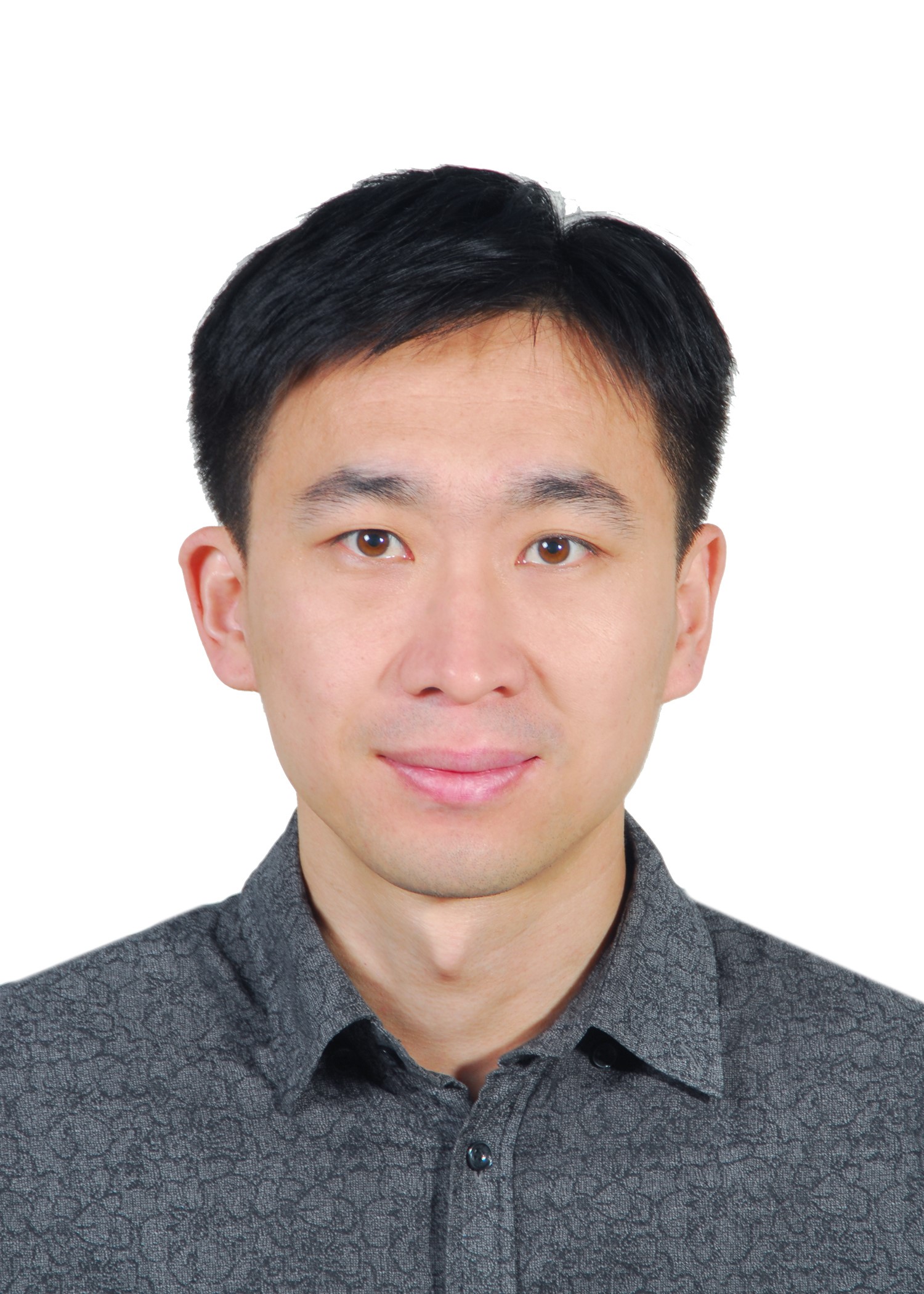}}]{Jingyu Wang} obtained his PhD degree from Beijing University of Posts and Telecommunications in 2008. He is currently a professor of State Key Laboratory of Networking and Switching Technology at Beijing University of Posts and Telecommunications. He has published more than 100 papers in international journal, including IEEE CMag, TVT, ISJ, TSC, TMM, TCC, IoT, TWC, and so on. His research interests span broad aspects of SDN/NFV, edge/cloud computing, IoV/IoT, big data processing and transmission, intelligent networks, and traffic engineering.
\end{IEEEbiography}
\vspace{-2em}

\begin{IEEEbiography}[{\includegraphics[width=1in,height=1.25in,clip,keepaspectratio]{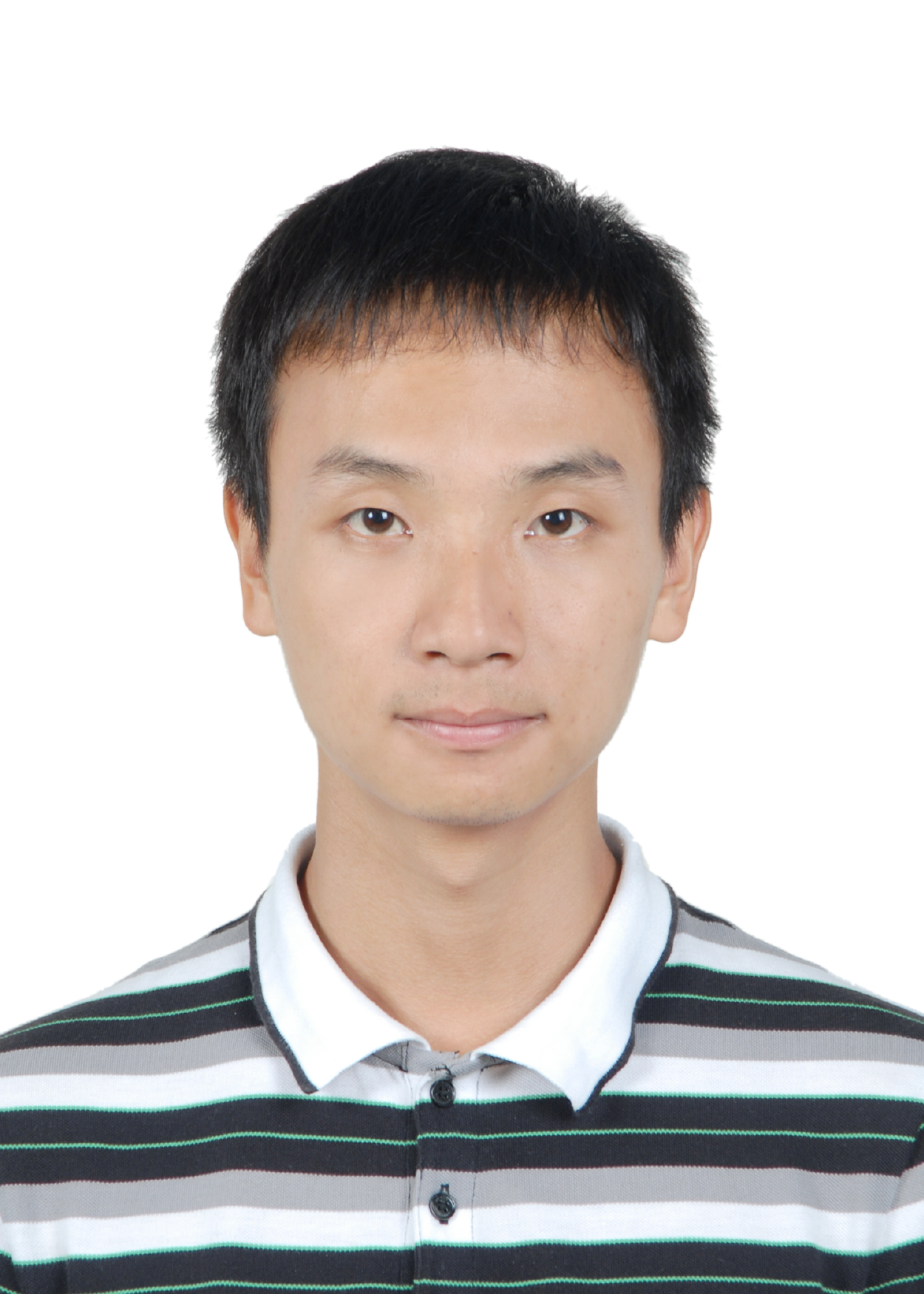}}]{Haifeng Sun}
	obtained his PhD degree from Beijing University of Posts and Telecommunications in 2017. He is currently a lecture of State Key Laboratory of Networking and Switching Technology at Beijing University of Posts and Telecommunications. His research interests span broad aspects of AI, NLP, big data analysis, object detection, deep learning, deep reinforcement learning, SDN, processing.
\end{IEEEbiography}
\vspace{-2em}

\begin{IEEEbiography}[{\includegraphics[width=1in,height=1.25in,clip,keepaspectratio]{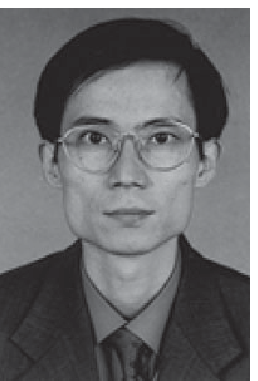}}]{Jianxin Liao}
	obtained his Ph.D degree at University of Electronics Science and Technology of China in 1996. He is currently the dean of Network Intelligence Research Center and the full professor of State Key laboratory of Networking and Switching Technology in Beijing University of Posts and Telecommunications. He has published hundreds of research papers and several books. His main research interests include cloud computing, mobile intelligent network, service network intelligent, networking architectures and protocols, and multimedia communication.
\end{IEEEbiography}

\vspace{-2em}

\begin{IEEEbiography}[{\includegraphics[width=1in,height=1.25in,clip,keepaspectratio]{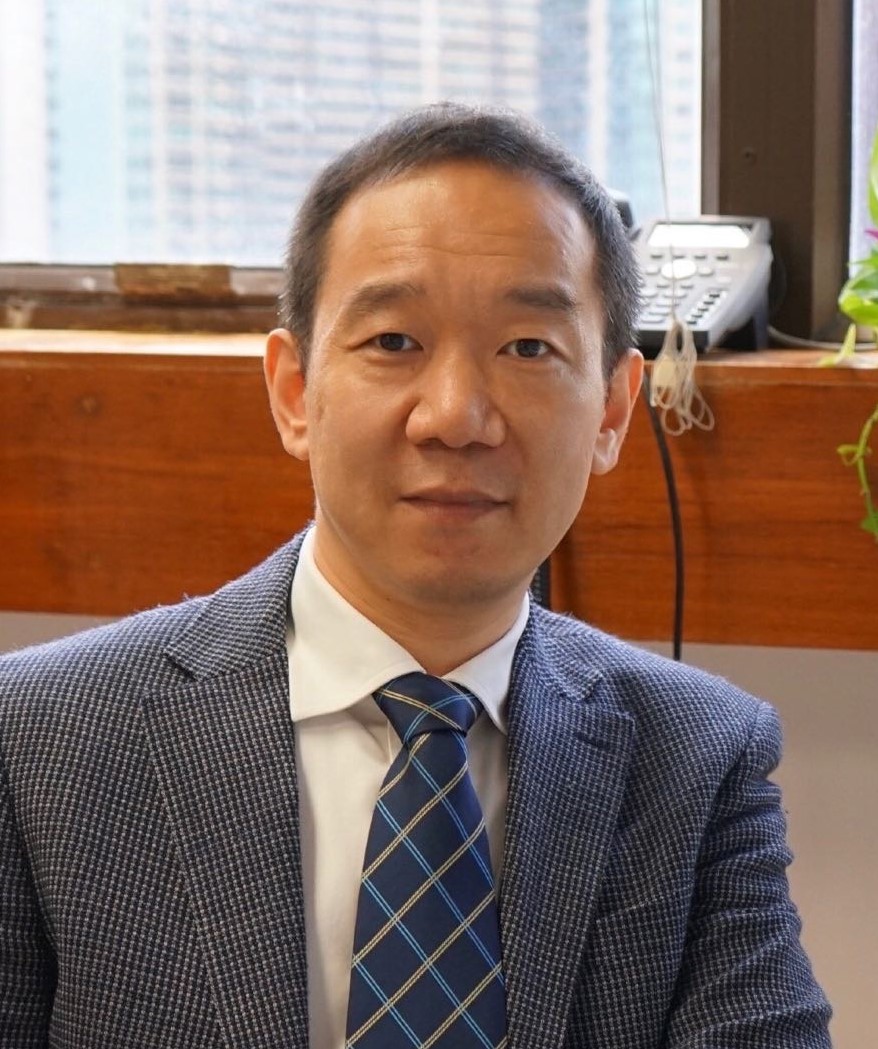}}]{Song Guo} is a Full Professor and Associate Head (Research \& Development) in the Department of Computing at The Hong Kong Polytechnic University. His research interests are mainly in edge AI, big data and machine learning, mobile computing, and distributed systems. He has served on IEEE Fellow Evaluation Committees for both CS and ComSoc, and been named on editorial board of a number of prestigious international journals like IEEE TC, IEEE TPDS, IEEE TCC, IEEE TETC, ACM CSUR, etc.
\end{IEEEbiography}

\end{document}